\newif\ifarxiv
\arxivtrue

\ifarxiv
    \documentclass[a4paper,USenglish,cleveref,autoref,thm-restate]{lipics-v2021} %
    \hideLIPIcs  %
    \nolinenumbers
\else
    \documentclass[a4paper,USenglish,cleveref,autoref,thm-restate]{lipics-v2021} %
\fi

\usepackage[skins]{tcolorbox}   %

\usepackage{xargs}  %

\usepackage{tikz}
\usetikzlibrary{positioning,fit,shapes,arrows.meta,decorations.pathreplacing, calc,decorations.text, decorations.markings}

\usepackage{thmtools,thm-restate}

\author{Joanne Dumont}{Univ. Orl\'eans, INSA Centre Val de Loire, LIFO UR 4022, F-45067 Orl\'eans, France}{}{}{}
\author{Edouard Nemery}{Universit\'{e} Paris-Dauphine, PSL University, CNRS UMR7243, LAMSADE, Paris, France}{}{}{}
\author{Anthony Perez}{Univ. Orl\'eans, INSA Centre Val de Loire, LIFO UR 4022, F-45067 Orl\'eans, France}{}{}{}
\author{Florian Sikora}{Universit\'{e} Paris-Dauphine, PSL University, CNRS UMR7243, LAMSADE, Paris, France}{}{}{}

\usepackage[commandRef=cref,createShortEnv, conf={normal, no link to proof}]{proof-at-the-end} 

\newenvironment{mainresult}[4]{%
    \begin{#1E}
    \label{#2}
    #3
    \end{#1E}
}

\newenvironment{apxresult}[4]{%
\ifarxiv
    \begin{#1E}
    \label{#2}
    #3
    \end{#1E}
\else
    \begin{#1E}[$\star$][end, restate, no link to proof, category=#4]
    \label{#2}
    #3
    \end{#1E}
\fi
}

\authorrunning{J. Dumont and E. Nemery and A. Perez and F. Sikora} %

\Copyright{Joanne Dumont and Edouard Nemery and Anthony Perez and Florian Sikora} %

\ccsdesc[500]{Mathematics of computing~Graph algorithms}

\keywords{Parameterized Complexity, Broadcast Problem, Structural Parameter} 

\relatedversion{} %

\acknowledgements{The authors thank Florent Foucaud and Prafullkumar Tale for discussions about the problem. We are also grateful to anonymous referees for insightful comments that really helped improve the presentation of this work. }%

\title{On the Parameterized Complexity of \BI{} and \BP}
\titlerunning{Parameterized Complexity of \BI{} and \BP} %

\renewcommand{\le}{\leqslant}
\renewcommand{\ge}{\geqslant}
\renewcommand{\leq}{\leqslant}
\renewcommand{\geq}{\geqslant}

\def\BI{\textsc{Broadcast Independence}}
\def\PBI{$p$-\textsc{Broadcast Independence}}
\def\WBI{\textsc{Weighted \BI{}}}
\def\WPBI{\textsc{Weighted} \PBI}
\def\BP{\textsc{Broadcast Packing}}
\def\PBP{$p$-\textsc{Broadcast Packing}}
\def\WBP{\textsc{Weighted Broadcast Packing}}
\def\WPBP{\textsc{Weighted} \PBP}

\newtcolorbox{mypb}[2][]
{
    enhanced,
    boxed title style = {colframe=white},
    attach boxed title to top left={
        xshift=0.5cm,
        yshift= -3.5mm,     
    },
    top=4mm,
    coltitle=black,
    beforeafter skip=\baselineskip,
    colframe = lightgray,
    colback  = white,
    colbacktitle  = white,
    coltitle = black,  
    fonttitle = \scshape,
    titlerule = 0mm, 
    title    = {#2},
    #1
}

\newcommand{\Pb}[3]{%
    \begin{mypb}{#1}
       \textbf{\textsf{Input}}: #2%
       \par\noindent%
       \textbf{\textsf{Output}}: #3%
       \smallskip%
       \par\noindent%
    \end{mypb}%
}

\newcommand{\bra}[1]{[#1]}
\newcommand{\braz}[1]{[#1]_0}

\newcommand{\brp}{\rho}
\newcommand{\bri}{\sigma}

\newcommand{\diam}{\operatorname{diam}}
\newcommand{\ecc}{\operatorname{ecc}}
\newcommand{\tw}{\operatorname{tw}}
\newcommand{\fvs}{\operatorname{fvs}}
\newcommand{\pw}{\operatorname{pw}}
\newcommand{\vc}{\operatorname{vc}}
\newcommand{\td}{\operatorname{td}}

\newcommand{\val}{value}
\newcommand{\vals}{values}

\newcommand{\ABP}[1]{#1^{\brp}} %
\newcommand{\ABI}[1]{#1^{\bri}} %

\usepackage{usebib}
\newbibfield{authortext}
\bibinput{bib}

\newcommand{\authorcite}[1]{\usebibentry{#1}{authortext}~\cite{#1}}

\usepackage{mathtools}
\newcommand\myeq{\stackrel{\mathclap{\mbox{def}}}{=}}

\begin{document}

\maketitle

\begin{abstract}
A broadcast on graphs models is a natural and well-studied way to capture limited-range transmissions. 
Formally, a broadcast on a connected graph $G = (V,E)$ is a function $f$ that assigns each vertex $v$ an integer $f(v)$ with $0 \leq f(v) \leq \ecc(v)$ 
where $\ecc(v)$ denotes the eccentricity of $v$. 
A vertex $u$ \emph{hears} a broadcasting vertex $v$ (with $f(v)>0$) if $u$ is at distance at most $f(v)$ from $v$. 
Beyond the classical broadcast domination problem, where every vertex is required to hear at least one vertex, two variants 
raise intriguing combinatorial and algorithmic questions. 
In an \emph{independent broadcast}, no broadcasting vertex hears another broadcasting vertex, while a \emph{broadcast packing} 
requires that every vertex hears at most one broadcasting vertex.
The corresponding 
problems \BI{} and \BP{} ask for broadcasts of \vals{} at least $k$ under these constraints, where the \val{} is the sum of the broadcast \vals{}. 
We initiate a systematic study of the parameterized complexity of such problems. 
We prove that \BI{} and \BP{} are FPT parameterized by the treewidth plus the diameter of $G$, with a family of dynamic-programming algorithms over nice tree decompositions. 
We obtain as a corollary that both problems are FPT parameterized by $k$ and the treewidth of $G$ and XP for treewidth only. 
The latter result shows that the known algorithm for trees (Bessy and Rautenbach, DAM 2022) can indeed be extended to bounded-treewidth graphs. 
On the negative side, we show that \BI{} is W[1]-hard parameterized by the pathwidth of $G$. 
Note that this result completes the picture for parameter $k$ and treewidth for \BI{} since it is known to be W[1]-hard for $k$ only.
We complement these results by showing that a weighted version of both problems, where the input comes with a weight function 
on the edges, is W[1]-hard parameterized by the vertex cover of $G$. 
Finally, we provide a constant-factor approximation algorithm parameterized by treewidth for \BI. 
\end{abstract}

\section{Introduction}
\label{sec:intro}

The concept of broadcast domination and broadcast independence has been introduced by Erwin in his PhD thesis~\cite{Erwin2001}. 
A broadcast of a connected graph $G = (V,E)$ is a function $f \colon V \to \{0, \ldots, \diam(G)\}$ 
such that $f(v) \leqslant \ecc(v)$ for every vertex $v$ of $V$, where $\ecc(v)$ denotes the eccentricity of $v$, i.e. the maximum distance from $v$ to any other vertex, and $\diam(G)$ its diameter, i.e. its maximum eccentricity.  
A vertex $v$ is \emph{broadcasting} if $f(v) > 0$ and  
a vertex $u$ \emph{hears} a broadcasting vertex $v$ if $u$ is at distance at most $f(v)$ from $v$ (in particular each broadcasting vertex hears itself).
The \val{} of a broadcast $f$ is $\sum_{v \in V}f(v)$. 
Finding a broadcast of minimum \val{} such that every vertex hears a broadcasting vertex is thus a natural variant of the classical \textsc{Dominating Set} problem. 
Surprisingly and in sharp contrast with \textsc{Dominating Set} which is one of Karp's 21 NP-complete problems~\cite{Karp1972}, a beautiful polynomial-time algorithm finding an optimal broadcast domination in polynomial time has been given by Heggernes and Lokshtanov~\cite{HL06}. 

\medskip

Several other broadcast variants have been studied in the literature, notably broadcast \emph{independence}, \emph{irredundance} and \emph{packing} (see for instance the introductory section of~\cite{DBLP:journals/dam/BouchouikaBS20} for more details). 
In this paper, we focus on independence and packing. 
A broadcast $\bri$ is \emph{independent}\footnote{In the remainder of the paper and to avoid confusion between the two notions we use $\bri$ to denote independent broadcasts and $\brp$ for broadcast packings.} if no broadcasting vertex hears another broadcasting vertex; more precisely, for any distinct broadcasting vertices $u,v$, we must have $d_G(u,v) > \max(\bri(u),\bri(v))$ where $d_G$ is the distance in $G$.
A broadcast $\brp$ is a \emph{packing}\footnotemark[1] if every vertex hears at most one broadcasting vertex; 
or, equivalently, for any distinct broadcasting vertices $u,v$, we must have $d_G(u,v) > \brp(u)+\brp(v)$. 
The corresponding optimization problems ask for a broadcast of maximum \val{} $\sum_{v \in V} f(v)$ that satisfies one of the aforementioned properties.
We denote these problems by \BI{} and \BP, respectively, and define them as follows in their decision version:

\Pb{\BI}
{a graph $G = (V,E)$, an integer $k \in \mathbb{N}$.}
{a broadcast $\bri$ of $G$ of \val{} at least $k$ such that for any $u,v \in V$ with $u \neq v$ and $\bri(u)>0, \bri(v)>0$, it holds that $d_G(u,v) > \max(\bri(u),\bri(v))$.}

\Pb{\BP}
{a graph $G = (V,E)$, an integer $k \in \mathbb{N}$.}
{a broadcast $\brp$ of $G$ of \val{} at least $k$ such that for any $u,v \in V$ with $u \neq v$ and $\brp(u) > 0$, $\brp(v) > 0$, it holds that $d_G(u,v) > \brp(u) + \brp(v)$.}

\subparagraph*{Related work} Although many structural bounds have been given in the last decade about both broadcast variants~\cite{DBLP:journals/dam/AhmaneBS18,bessy2018girth,DBLP:journals/dm/BessyR19,DBLP:journals/dam/BouchouikaBS20,brewster2024broadcast,DBLP:journals/dam/DunbarEHHH06,hoepner2023boundary,McDonald25}  
the complexity status remained open for quite some time even on simple graph classes. 
For instance, the complexity status of \BI{} on trees was repeatedly asked for in the literature~\cite{DBLP:journals/dam/AhmaneBS18,DBLP:journals/dam/DunbarEHHH06,hedetniemi2006unsolved}  
and was settled only recently by \authorcite{DBLP:journals/dam/BessyR22} who gave an $O(n^9)$ time algorithm. 
This result was subsequently improved to $O(n^8)$ in~\cite{McDonald25}. 
Algorithms for other specific subclasses of trees (such as the one for locally uniform 2-lobster of fixed length) have also been developed, see e.g.~\cite{DBLP:journals/dam/AhmaneBS18,DBLP:journals/dmgt/AhmaneBS24}.  
Regarding the general case there is little hope 
to have a polynomial-time algorithm as Bessy and Rautenbach proved that \BI{} is NP-hard to approximate within $n^{1-\varepsilon}$ on general graphs, and also NP-hard on planar graphs with maximum degree $4$~\cite{DBLP:journals/dam/BessyR22}. 
We note that the authors of~\cite{DBLP:journals/dam/BessyR22} left open the possibility to adapt their algorithm on trees to graphs of bounded treewidth. 
The \BI{} problem shares some similarity with the \textsc{Distance $p$-independent set}~\cite{eto2014distance,montealegre2016distance} (also called \textsc{$p$-scattered set} in~\cite{DBLP:journals/dam/KatsikarelisLP22}) in which 
one seeks an independent broadcast of maximum \val{} %
where every vertex broadcasts at distance exactly $p$. Notably, it is FPT parameterized by treewidth for fixed $p$, but becomes $W[1]$-hard when $p$ is given in the input~\cite{DBLP:journals/dam/KatsikarelisLP22}.

\subparagraph*{Our results} We follow this line of research by initiating a study of \BI{} and \BP{} from the parameterized complexity viewpoint. 
In such a context the problem comes together with some parameter $k \in \mathbb{N}$ %
and the aim is to decide the problem in time $f(k) \cdot poly(n)$ where $f$ is any computable function and $n$ is the size of the instance. If such an algorithm exists then the problem is Fixed-Parameter Tractable (FPT for short). 
It is possible to show that a problem is unlikely to be fixed-parameter tractable by proving W[1]-hardness via parameterized reductions.
In such a case one seeks an XP algorithm with complexity $f(k) \cdot n^{g(k)}$, where $f$ and $g$ are computable functions. %
We refer the reader to \cite{DBLP:books/sp/CyganFKLMPPS15} for precise definitions of the W-classes. 
A natural parameterization for both problems is the \val{} $k$ required for the sought broadcast. 
We moreover consider several \emph{structural} parameterizations of the problems by mainly considering the well-known treewidth ($\tw$), pathwidth ($\pw$), vertex cover ($\vc$), and feedback vertex set ($\fvs$) parameters. We refer the reader to~\cite{DBLP:books/sp/CyganFKLMPPS15} for more definitions regarding parameterized complexity. %

\medskip

Our contributions are both positive (with algorithms) and negative (with parameterized hardness), see also \cref{fig:bi_complexity} for an illustration. 
Below we state our main results for both problems. 

\begin{figure}
    \centering
    \begin{tikzpicture}[
	param/.style={draw, rounded corners, minimum width=2.6cm, minimum height=0.9cm, align=center},
	arrow/.style={->, thick},
	scale=0.8, transform shape
	]
	
	\node[param,fill=green!50] (vc)  at (2,0)   {Vertex Cover\\ \cref{cor:bi_fpt_vc_td}};
	
	\node[param,fill=orange!50] (fvs) at (0,3)   {Feedback Vertex Set \\\cref{thm:BIFVSPWhard}};
	\node[param,fill=green!50] (td)  at (4,2)   {Treedepth\\\cref{cor:bi_fpt_vc_td}};
	
	\node[param,fill=orange!50] (pw)  at (4,4)   {Pathwidth\\\cref{thm:BIFVSPWhard}};
	
	\node[param,fill=green!50] (twk)  at (8,4)   {Treewidth + $k$\\\cref{cor:BIBPtwkfpt}};
	
	\node[param,fill=orange!50] (k)  at (8,6)   {$k$\\\cite[Proposition 1.1]{DBLP:journals/dam/BessyR22}};
	
	\node[param,fill=green!50] (twdiam)  at (12,4)   {Treewidth + diameter\\ \cref{thm:bitw}};
	
	\node[param,fill=red!50] (diam)  at (12,6)   {Diameter\\ \cite{DBLP:journals/dam/BessyR22}};
	
	\node[param,fill=orange!50] (tw)  at (3,6)   {Treewidth\\\cref{thm:BIFVSPWhard}};
	
	\node[param,fill=red!50] (dtp) at (-1,6) {Distance to\\ Planar \cite{DBLP:journals/dam/BessyR22}};

	\draw[arrow] (vc) -- (fvs);
	\draw[arrow] (vc) -- (td);
	\draw[arrow] (td) -- (pw);
	\draw[arrow] (pw) -- (tw);
	\draw[arrow] (fvs) -- (tw);
	\draw[arrow] (fvs) -- (dtp);
	\draw[arrow] (twk) -- (tw);
	\draw[arrow] (twk) -- (k);
	\draw[arrow] (twdiam) -- (tw);
	\draw[arrow] (twdiam) -- (diam);
	\draw[arrow] (td) -- (twdiam);
\end{tikzpicture}
    \caption{Complexity of \BI\ parameterized by several parameters. In red: \BI\ is NP-hard even when this parameter has fixed value. In orange: \BI\ is W[1]-hard for this parameter. In green: \BI\ admits an FPT algorithm for this parameter. Note that for boxes in orange,  
    XP algorithms are known (for parameter $k$) or provided in this work. 
    An arrow from parameter $x$ to parameter $y$ means that there is a function $f$ such that $y(G) \leq f(x(G))$ for any given graph $G$.
    Hence, positive results propagate downwards, while negative results propagate upwards.
    }
    \label{fig:bi_complexity}
\end{figure}
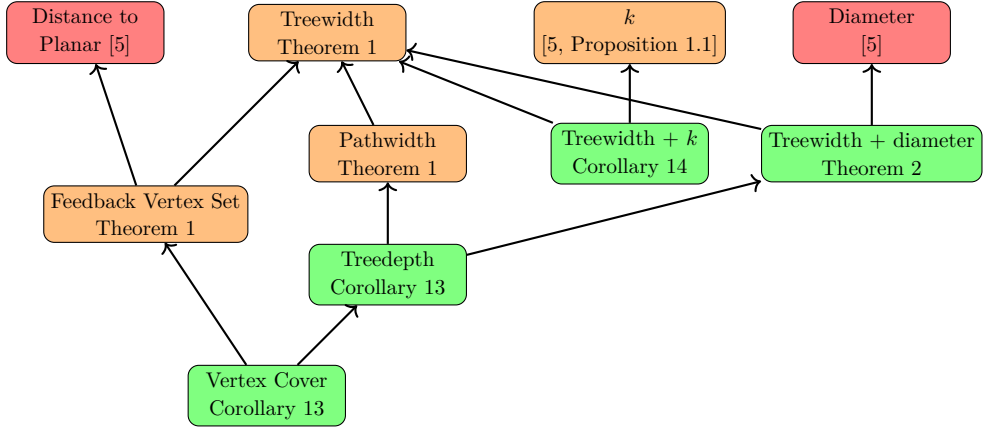

\begin{restatable}{theorem}{ThmBIFVSPWhard}
\label{thm:BIFVSPWhard}
    \BI{} is W[1]-hard when parameterized by $\fvs+\pw$. 
\end{restatable}

\begin{restatable}{theorem}{ThmBItw}
\label{thm:bitw}
    \BI{} can be solved in time $O^*(\diam(G)^{O(\tw)})$\footnote{We use the $O^*$ notation to indicate that polynomial factors are omitted from the expression. That is, we write $O^*(f(n))$  whenever the complexity is $O\left(f(n) \cdot poly(n)\right)$.} on graphs of treewidth $\tw$ given a nice tree-decomposition of $G$. 
\end{restatable} 

As we shall see in \cref{sec:progdyn}, \cref{thm:bitw} is obtained through a variant of \BI. 
Since the diameter of a graph is bounded by its number of vertices, \cref{thm:bitw} yields an XP algorithm parameterized by treewidth with running time $n^{O(\tw)}$. This confirms that the algorithm of \authorcite{DBLP:journals/dam/BessyR22} on trees can indeed be extended to bounded-treewidth graphs. %
\medskip

In weighted variants of both problems, namely \WBI{} and \WBP, the input comes with a weight function $\omega \colon E \rightarrow \mathbb{N}$ representing the \emph{length} of each edge. %
In this setting, the distances (and diameters) are calculated relatively to those weights in the classical way. 
We hereafter work with distances as stated in the formal definitions of both (unweighted) problems.
Note that in the case of \WBP, the equivalence between stating that any vertex hears at most one vertex and stating that $d_G(u,v) > f(u) + f(v)$ for any distinct broadcasting vertices $u$ and $v$ does not hold. 

\begin{restatable}{theorem}{ThmWBPVCHard}
\label{thm:WBPVChard}
    \WBP{} parameterized by vertex cover is W[1]-hard even when the weights are polynomially bounded.
\end{restatable}

We prove a similar result for \WBI{} which is in turn used to establish \cref{thm:BIFVSPWhard}. 
In contrast, we were not able to adapt this reduction to obtain a similar result for \WBP{}.  %
However, the algorithm used in \cref{thm:bitw} can be adapted to this setting, leading to the following result. 

\begin{restatable}{theorem}{ThmBPtw}
\label{thm:bptw}
    \BP{} can be solved in time $O^*(\diam(G)^{O(\tw)})$ on graphs of treewidth $\tw$ given a nice tree-decomposition of $G$. 
\end{restatable}

Our FPT algorithms are obtained by providing dynamic programming on nice tree-decompositions of the input graph (see \cref{def:td,def:ntd}) relying on compact signatures, 
while our hardness proofs are inspired by a reduction of \authorcite{DBLP:journals/dam/KatsikarelisLP22} for the \textsc{$d$-scattered set} problem. 
We moreover observe that both \BI{} and \BP{} are FPT for parameter $k + \tw(G)$. This is a rather straightforward corollary from \cref{thm:bitw,thm:bptw}.  
Note that since \BI{} is W[1]-hard parameterized by $k$ only\footnote{There is a simple reduction from the W[1]-hard \textsc{Maximum Independent Set} problem with a universal vertex (see e.g.~\cite[Prop. 1.1]{DBLP:journals/dam/BessyR22}).}  
and by $\tw$ only (due to \cref{thm:BIFVSPWhard}), our results complete the picture for such parameters. 
To conclude, we provide a constant-factor approximation algorithm parameterized by treewidth for \BI. 

\begin{restatable}{theorem}{ThmBIapprox}
\label{thm:fptapproxbi}
    \BI{} can be approximated within a ratio of $\frac{1}{2}-\epsilon$ in time $O((\frac{1}{\epsilon})^{\tw}n^{O(1)})$ for all $\epsilon>0$.
\end{restatable}

Along the way we provide a structural result that generalizes a known result by \authorcite{DBLP:journals/dm/BessyR19} who showed 
that the size of a maximum independent broadcast is at most $4$ times the size of a maximum independent set. 

\subparagraph*{Outline} \cref{sec:prelim} presents some preliminary notions for graphs and broadcast functions used in this paper. 
\cref{sec:progdyn} describes our dynamic programming algorithms for \cref{thm:bitw,thm:bptw} while \cref{sec:hardness} presents our W[1]-hardness reductions. Finally, \cref{sec:approx} presents a constant-factor approximation algorithm for \BI{} parameterized by treewidth. 

\section{Preliminaries}
\label{sec:prelim}
We consider undirected finite connected graphs, denoted $G=(V,E)$ where $V(G)$ is the vertex set of $G$ and $E(G)$ its edge set. For the sake of readability we use $V$ and $E$ whenever the context is clear. 
Throughout the paper we let $n := \vert V \vert$ and $m := \vert E \vert$. We sometimes consider 
weighted graphs, with a weight function $\omega \colon E \rightarrow \mathbb{N}$ defined over the edges of $G$.  
The (weighted) shortest distance between two vertices $u$ and $v$ is denoted by $d_G(u,v)$ (we drop the subscript whenever the context is clear). 
The eccentricity of a vertex $v$ is denoted $\ecc(v)$ and is the largest distance between $v$ and any other vertex of $G$. 
The diameter of a graph $G$ is denoted $\diam(G)$ and is the maximum eccentricity among vertices (note that since we consider connected graphs the diameter will always be finite).
Given any integer $n$, we let \bra{$n$} denote the set $\{1, \ldots, n\}$ and $\braz{n} = \{0\} \cup \bra{n}$. %
A function 
$f : V \to \braz{\diam(G)}$
is a \emph{broadcast} if $f(v) \leq \ecc(v)$ for every vertex $v$. A vertex $v$ is \emph{broadcasting} (or a \emph{broadcaster}) if $f(v) > 0$. 
The \emph{\val{}} of $f$ is $\sum_{v \in V} f (v)$. %
A vertex $u$ \emph{hears} $v$ if $d(u, v) \leq f(v)$. 
For a set $X \subseteq V$, we denote by $X^f := \{v \in X \mid f(v)>0 \}$ the set of broadcasting vertices in $X$. 
A broadcast is \emph{independent} if $d_G(u, v) > \max(f(u), f(v))$ for any $u$ and $v$ broadcasting vertices, i.e. no broadcasting vertex hears another broadcasting vertex. Similarly, a broadcast is \emph{packing}
if $d_G(u, v) > f(u) + f(v)$ for any broadcasting vertices $u$ and $v$, 
i.e. no vertex hears more than one vertex. 
Recall that in edge-weighted variants of both problems we consider the distance-based definitions of \emph{independent} and \emph{packing}.%
Given any broadcast $f$, we denote $N_f[v] = \{u \mid d(u,v) \leq f(v)\}$ the \emph{broadcast neighborhood} of $v$ (we can see it as a ball of center $v$ and radius $f(v)$). 
Given a graph $G=(V,E)$ and a broadcast $f$ of $G$, for all $X\subseteq V$ we note $f(X):=\sum_{v \in X} f (v)$. 

\medskip

We also consider a variant for both problems where the \val{} of any broadcasting vertex is bounded by some fixed integer $p$. We call such broadcasts $p$-broadcasts and such problems \PBI{} and \PBP, respectively. 
Note that we can assume that $p\leq \diam(G)$ and that $1$-\BI{} is equivalent to the independence number problem. 

\medskip

The following lemma shows that, for \BI{} and \BP{}, we can safely ignore the per-vertex eccentricity constraints and instead allow broadcast \vals{} to range up to the diameter of the graph.
We will consider this version of the definition in the rest of the paper, as it allows for simplifications in the reductions and in the design of algorithms.

\apxresult{lemma}{lem:bounddiam}
{
    Let $f:V\to \braz{\diam(G)}$ be any function. The following holds:
    \begin{enumerate}
        \item\label{bounddiamiii} if there exists a unique $u$ such that $f(u)>0$, then there exists an independent broadcast and a broadcast packing of \val{} at least $\sum_{v \in V} f(v) = f(u)$.
        \item\label{bounddiami} if $d(u,v) > \max (f(u),f(v))$ for any two vertices $u$ and $v$ such that $f(u)>0$ and $f(v)>0$,  
        then $f$ is an independent broadcast. 
        \item\label{bounddiamii} if $d(u,v) > f(u) + f(v)$ for any two vertices $u$ and $v$ such that $f(u)>0$ and $f(v)>0$, 
        then $f$ is a broadcast packing.
    \end{enumerate}
}
{prelim}

\begin{proofE}
    Let $G = (V,E)$ be a graph and $f:V\rightarrow \braz{\diam(G)}$ any function.  
    First, assume that $f$ has exactly one vertex with $f(v)>0$. In that case, let $v\in V(G)$ be such that $\ecc(v)=\diam(G)$ and $f'$ be the function defined by: 
    $f'(u):=\begin{cases}
        \diam(G) &\text{ if} \ u=v \\
        0 & \text{ otherwise} \\
    \end{cases}$.
    
    Then by construction $f'$ is an independent broadcast (and a broadcast packing) of \val{} superior or equal to the one of $f$, which proves \cref{bounddiamiii}.
    
    \medskip
    
    Otherwise, assume that $f$ respects the prerequisites of \cref{bounddiami} and let $u,v$, $u\neq v$, be two vertices such that $f(u)>0$ and $f(v)>0$. 
    Then, by assumption, $f(u)< d(u,v)\leq \ecc(u)$, implying that $f$ is a broadcast, and in particular an independent one. 
    Note that if \cref{bounddiamii} holds then \cref{bounddiami} holds and hence it is also a broadcast packing. 
\end{proofE}

\subparagraph*{Graph decompositions} 
We recall standard definitions of tree-decompositions that will be used in \cref{thm:bitw,thm:bptw}.

\begin{definition}[Tree Decomposition]%
\label{def:td}
A \emph{tree decomposition} of a graph $G = (V,E)$ is a pair $\mathcal{T} = (T,\{B_x\}_{x\in V(T)})$ in
which each vertex $x \in T$ has an assigned set of vertices $B_x \subseteq V$ (called a \emph{bag}) such that $\bigcup_{x \in T} B_x = V$ with the following properties:
\begin{itemize}
    \item for any $uv \in E$, there exists an $x \in V(T)$ such that $u, v \in B_x$.
    \item if $v \in B_x$ and $v \in B_y$, then $v \in B_z$ for all $z$ on the path from $x$ to $y$ in $T$
\end{itemize}
\end{definition}

\begin{definition}[Nice Tree Decomposition \cite{cygan2011solving}]
\label{def:ntd}
A \emph{nice tree decomposition} is a tree decomposition $\mathcal{T} = (T,\{B_x\}_{x\in V(T)})$ where $T$ is a rooted tree, with a special bag $z$ called the root with $B_z = \emptyset$ and in which each bag is of one of the following types:
\begin{itemize}
    \item \textbf{Leaf bag}: a leaf $x$ of $T$ with $B_x = \emptyset$.
    \item \textbf{Introduce vertex bag}: an internal vertex $x$ of $T$ with one child vertex $y$ for which $B_x = B_y \cup \{v\}$ holds for some $v \notin B_y$. This bag is said to \emph{introduce} $v$.
    \item \textbf{Forget bag}: an internal vertex $x$ of $T$ with one child bag $y$ for which $B_x = B_y \setminus \{v\}$ holds for some $v \in B_y$. This bag is said to \emph{forget} $v$.
    \item \textbf{Join bag}: an internal vertex $x$ with two child vertices $y$ and $z$ with $B_x = B_y = B_z$.
\end{itemize}
Moreover, given a tree decomposition, one can compute a nice tree decomposition of the same width in polynomial time.
\end{definition}

\subparagraph*{Structural parameters~\cite{DBLP:books/sp/CyganFKLMPPS15,Hierarchy}} 
The \emph{width} of a tree-decomposition of $G$ is the maximum size of any of its bags minus $1$, and 
\emph{treewidth} of $G$, denoted $\tw(G)$ (or $\tw$ for short) is the minimum width of a tree decomposition of $G$. 
The \emph{pathwidth} $\pw$ of $G$ is defined analogously by requiring that $T$ is a path in \cref{def:td}. 
The \emph{feedback vertex set} number $\fvs$ of $G$ is the minimum number of vertices to remove from $G$ to obtain a forest while its \emph{vertex cover} number $\vc$ is the minimum number of vertices to be removed to obtain an edgeless graph. %

\section{Parameterized algorithms for treewidth and diameter}
\label{sec:progdyn}

\subparagraph*{Notations} %
We consider a graph $G = (V,E)$ together with a tree decomposition of $G$ of width $\tw$. 
We will describe a dynamic programming algorithm over a nice tree decomposition $(T,\{B_x\}_{x\in V(T)})$ of $G$ of width $\tw$. 
Recall that such a nice tree decomposition can be computed in polynomial time from the one for $G$ (\cref{def:ntd}).  %
For each node $x\in V(T)$, let $V_x\subseteq V$ denote the set of all vertices of $G$ that appear in bags of the subtree $T_x$ rooted at $x$, excluding those in $B_x$, i.e. vertices which are forgotten for $x$. 
More formally, we let $V_x := \bigl(\bigcup_{t \in T_x} B_t\bigr) \setminus B_x$ %
and write $G_x:=G[V_x]$. Let $f_x$ be a broadcast of $G_x$. 
The \emph{natural extension} of $f_x$ to $G$ is the broadcast $f$ defined by: \[
    f(v) = 
    \begin{cases} 
        f_x(v) & \text{if } v\in V_x,\\ 
        0       & \text{otherwise.}
    \end{cases}
\]
For the sake of simplicity, we do not distinguish between a broadcast of an intermediate graph $G_x$ and its natural extension to the final graph $G$. 
Finally, recall that $V_x^f$ denotes the set of broadcasting vertices in $V_x$ (with respect to $f$). 
In the following proofs considering edge-weighted graphs, we will not refer explicitly to the weights of edges but rather to the distance function $d$ (or $d_G)$ that encompasses such weights. %
Such distances can be pre-computed beforehand in polynomial-time. 

We first turn our attention to proving \cref{thm:bitw} which we recall here for the sake of readability. 

\ThmBItw*

We actually prove a slightly more general result stated in the following theorem. 

\apxresult{theorem}{thm:WBPItw}
{
    The \WPBI{} problem is FPT when parameterized by the treewidth $\tw$ of the input graph. 
    In particular, there is a dynamic programming algorithm that, given a nice tree-decomposition of $G$ of width $\tw$, 
    solves the problem in time $O^*((p+1)^{O(\tw)})$. 
}
{progdyn}

\begin{proofE}
Recall that one can compute a nice tree-decomposition of $G$ of width $\tw$ in polynomial time (see \cref{def:ntd}). %
In the following proof, a broadcast is \emph{independent} if its natural extension is an independent broadcast in graph $G$. 
We define $\mathcal{B}(G_x)$ as the set of all broadcasts of $G_x$ and $\mathcal{B}^{in}(G_x)$ the set of all independent broadcasts of $G_x$.

Given a node $x \in V(T)$, a \emph{partial signature} is a function $s : B_x \to %
\braz{p}$. 
We assign to every node couples %
of partial signatures $(s_1, s_2)$ that will each provide a compact representation of how a partial broadcast over $G_x$ affects the vertices of $B_x$. 
We call such a couple of partial signatures a \emph{signature} and denote it by $s = (s_1, s_2)$.  We call $A_x$ the set of all signatures of $x$. 
Note that for each bag, there are at most $(p+1)^{2|B_x|}\leq (p+1)^{2\tw+2}$ different signatures.

\begin{definition}[Signature of a broadcast]
Let $x\in V(T)$ and $\bri\in \mathcal{B}(G_x)$. Let $s=(s_1,s_2)\in A_x$ be defined by, for all $v\in B_x$:
\[
    \begin{cases}
    s_1(v)=\max(0,\max\limits_{u\in V_x^\bri}(\bri(u)-d(u,v)+1)) \\
    s_2(v)=\min(p,\min\limits_{u\in V_x^\bri}(d(u,v)-1))
    \end{cases}
\]

Such a signature $s$ is called the \emph{signature of $\bri$ with respect to $x$}. We define the function $sgn_x \colon \mathcal{B}(G_x)\to A_x$ that assigns to each broadcast its signature with respect to $x$. 
\end{definition}

Intuitively, a signature $s=(s_1,s_2)$ encodes two pieces of information for every vertex $v\in B_x$: 
$s_1(v)$ is the smallest distance from v to any vertex that does not hear the same broadcaster as v (note that in some special cases this vertex might not actually exist)
while $s_2(v)$ represents the distance to the nearest broadcaster of $G_x$ (up to $p$). 
Let us recall that a vertex $v \in B_x$ does not belong to $G_x$. Hence note that $s_1(v) = 0$ mean that $v$ does not hear any vertex in (the natural extension of) $\bri$. 

\medskip

Let us now describe the Dynamic Programming table. We associate a \val{} to each signature $s$ of a node $x$ corresponding to the maximum \val{} of an independent broadcast with signature $s$, if it exists. Otherwise, we set the \val{} to $-\infty$. In the following, we aim to compute in a bottom-up fashion the \val{} of the only signature associated to $r$, the root of the tree decomposition. 
Formally, for all $x\in V(T)$, we define the function $val_x$ as follows:
\begin{align*}
    val_x\colon A_x & \to \mathbb{N} \\ 
    s & \mapsto \max\bigl\{ \bri(V_x)\ \big|\ \bri \in \mathcal{B}^{in}(G_x),\  sgn_x(\bri)=s \bigr\}
\end{align*}

If $sgn_x^{-1}\cap \mathcal{B}^{in}(G_x)=\emptyset$, we set $val_x(s)=-\infty$.

\medskip

We note that during the Dynamic Programming algorithm, an introduce node \emph{does not} decide whether the introduced vertex broadcasts or not, this decision will be taken at the corresponding forget node. 
Let $r$ be the root of $T$. By definition of a nice tree decomposition, $G_r=G$, and so, $\mathcal{B}^{in}(G_r)$ is exactly the set of all the independent broadcasts of $G$. 
Also, as $B_r=\emptyset$, there is only one signature $s\in A_r$. 
Therefore, $\mathcal{B}^{in}(G_r)\subseteq sgn_r^{-1}(s)$ and $val_r(s)$ is the \val{} of the best independent broadcast of $G$. Hence, if we have an efficient way to compute the \vals{} of $val_r$, then we can efficiently compute the solution to our problem. We now explain how to efficiently compute the \vals{} of $val_x$ for each type of node $x$. Take $x\in V(T)$.

\subparagraph*{Leaf node} If $x$ is a leaf node, then $B_x=\emptyset$ and $G_x$ is the empty graph. There is a single signature $s\in A_x$, and only one independent broadcast $\bri\in \mathcal{B}^{in}(G_x)$: the empty broadcast (of \val{} $0$). Hence $sgn_x(\bri)=s$ and $val_x(s)=0$.
Regarding the complexity, $val_x$ can be computed in constant time.

\subparagraph*{Introduce node} Suppose $x$ is an introduce node, and that its unique child is $y$. Let $v$ be the vertex of $V$ introduced such that $B_x=B_y\cup\{v\}$. Take a broadcast $\bri\in \mathcal{B}(G_x)$ and let $s=(s_1,s_2)=sgn_x(\bri)$. 
By definition of a tree decomposition, any path from 
 a vertex of $G_x$ to $v$ must go through the separator $B_y$. Therefore: 
 \begin{align*}
     s_1(v) & \myeq \max(0,\max_{u\in V_x^\bri}(\bri(u)-d(u,v)+1))\\
     & =\max(0,\max_{u\in V_x^\bri}(\max_{w\in B_y}(\bri(u)-d(u,w)-d(w,v)+1)))\\
     & =\max(0,\max_{w\in B_y}(s_1(w)-d(w,v)))
 \end{align*}
 and:
 \begin{align*}
     s_2(v) & \myeq \min(p,\min_{u\in V_x^\bri}(d(u,v)-1))\\
     & =\min(p,\min_{u\in V_x^\bri}(\min_{w\in B_y}(d(u,w) + d(w,v)-1)))\\
     & =\min(p,\min_{w\in B_y}(s_2(w)+d(w,v) )).
 \end{align*}
 
Let $s=(s_1,s_2)$ be a signature of $x$. If $s_1(v)=\max(0,\max_{u\in B_y}(s_1(u)-d(u,v)))$ and $s_2(v)=\min(p,\min_{u\in B_y}(s_2(u)+d(u,v)))$, the set of independent broadcasts having $s$ as a signature with respect to $x$ is exactly the set of independent broadcasts having $(s_1|_{B_y},s_2|_{B_y})$ as a signature with respect to $y$. Formally, $sgn_x^{-1}(s)\cap \mathcal{B}^{in}(G_x)=sgn_y^{-1}((s_1|_{B_y},s_2|_{B_y}))\cap\mathcal{B}^{in}(G_y)$. Otherwise, no independent broadcast has $s$ as a signature. For all signatures $s=(s_1,s_2)$ of $x$, we can deduce the following recursive formula:
\[
    val_x(s)=\begin{cases}
        val_y((s_1|_{B_y},s_2|_{B_y})) &\ \text{if}\ s_1(v)=\max(0,\max\limits_{u\in B_y}(s_1(u)-d(u,v)))\\ & \text{and}\ s_2(v)=\min(p,\min\limits_{u\in B_y}(s_2(u)+d(u,v))) \\
        -\infty & \text{otherwise}\\
    \end{cases}
\]

Regarding the complexity, assuming $val_y$ has already been computed, each \val{} of $val_x$ can be computed in time $O(|B_y|)=O(\tw)$, therefore, $val_x$ can be computed $O(\tw \cdot |A_x|)=O(\tw \cdot (p+1)^{2tw+2})$. 

\subparagraph*{Forget node} Suppose $x$ is a forget node and that its unique child is $y$. Let $v$ be a vertex of $V$ such that $B_y=B_x\cup \{v\}$. Take a broadcast $\bri\in \mathcal{B}(G_x)$, let $s=(s_1,s_2)=sgn_x(\bri)$, and $s'=(s'_1,s'_2)=sgn_y(\bri|_{V_y})$. If $\bri(v)=0$, then for all $u\in B_x$, $s_1(u)=s'_1(u)$ and $s_2(u)=s'_2(u)$. Otherwise, for all $u\in B_x$, it holds that:
\begin{align*}
    s_1(u)& \myeq \max\Bigl(0,\max_{\substack{w\in \ABI{V_x}}}(\bri(w)-d_G(u,w)+1)\Bigr) \\
    &=\max\Bigl(0, \max_{\substack{w \in \ABI{V_y}}} (\bri(w)-d_G(u,w))+1,\bri(v)-d_G(u,v)+1\Bigr)\\
    &=\max\Bigl(s'_1(u),\bri(v)-d_G(u,v)+1\Bigr)
\end{align*}

and:
\begin{align*}
    s_2(u)& \myeq \min\Bigl(p,\min_{\substack{w\in \ABI{V_x}}}(d_G(u,w))-1\Bigr) \\
    &=\min\Bigl(p, \min_{\substack{w \in \ABI{V_y}}} (d_G(u,w))-1,d_G(u,v)-1\Bigr)\\
    &=\min\Bigl(s'_2(u),d_G(u,v)-1\Bigr).
\end{align*}

Furthermore, $\bri\in \mathcal{B}^{in}(G_x)$ if and only if $\bri|_{V_y}\in \mathcal{B}^{in}(G_y)$ and either $\bri(v)=0$ or for all $w\in \ABI{V_y}$, $\bri(w)-d(v,w)<0$ and $\bri(v)-d(v,w)<0$, i.e., $s'_1(v)=0$ and $\bri(v)\leq s'_2(v)$. %
The set of independent broadcasts $\bri$ having $s$ as a signature is then exactly the set of broadcasts which verify one of the following properties: 
\begin{itemize}
    \item $\bri|_{V_y}\in \mathcal{B}^{in}(G_y)$, $\bri(v)=0$ and let $s'=(s'_1,s'_2)=sgn_y(\bri|_{V_y})$, where $(s'_1|_{B_x},s'_2|_{B_x})=s$.
    \item $\bri|_{V_y}\in \mathcal{B}^{in}(G_y)$, $\bri(v)>0$, and let $s'=(s'_1,s'_2)=sgn_y(\bri|_{V_y})$, where $s'_1(v)=0$ and $s'_2(v)\geq \bri(v)$.
\end{itemize}

For all signatures $s=(s_1,s_2)\in A_x$, we can deduce the following recursive formula:
\[
    val_x(s)=\max 
    \begin{cases}
        val_y((s'_1,s'_2)) & \mid s=(s'_1|_{B_x},s'_2|_{B_x}) \\
        val_y((s'_1,s'_2)) + l & \mid s'_1(v)=0,\ 1\leq l\leq s'_2(v) \text{ and for all } u\in B_x, \\
        & ~~s_1(u)=\max(s'_1(u), l- d(u,v)+1)\\ & ~~\text{ and }s_2(u)=\min(s'_2(u), d(u,v)-1).
    \end{cases}
\]

Regarding the complexity, assuming $val_y$ has already been computed, each \val{} of $val_y$ is used in the calculation of at most $p+1$ different \vals{} of $val_x$. %
Therefore, the complexity of this step is $O(p \cdot |A_y|)=O(p \cdot (p+1)^{2\tw+2})$.

\subparagraph*{Join Node} Suppose $x$ is a join node, and that its two children are $y$ and $z$. 
Let $\bri\in \mathcal{B}(G_x)$, $s'=(s'_1,s'_2)=sgn_{y}(\bri|_{V_{y}})$, $s''=(s''_1,s''_2)=sgn_{z}(\bri|_{V_{z}})$ and $s=(s_1,s_2)=sgn_x(\bri)$. As $V_x=V_{y}\biguplus V_{z}$, for all $v\in B_x$,
\begin{align*}
 s_1(v) & \myeq \max\bigl(0, \max_{\substack{u\in \ABI{V_x}}} (\bri(u) - d_G(v,u)+1)\bigr)\\
     &= \max\Bigl( 
            \max\bigl(0, \max_{\substack{u\in \ABI{V_{y}}}} (\bri(u)-d_G(v,u)+1)\bigr),
            \max\bigl(0, \max_{\substack{u\in \ABI{V_{z}}}} (\bri(u)-d_G(v,u)+1)\bigr)\Bigr) \\
    &= \max\bigl(s'_1(v),s''_1(v)\bigr)
\end{align*}

and:
\begin{align*}
 s_2(v) & \myeq \min\bigl(p, \min_{\substack{u\in \ABI{V_x}}} (d_G(v,u)-1)\bigr)\\
     &= \min\Bigl( 
            \min\bigl(p, \min_{\substack{u\in \ABI{V_{y}}}} (d_G(v,u)-1)\bigr),
            \min\bigl(p, \min_{\substack{u\in \ABI{V_{z}}}} (d_G(v,u)-1)\bigr)\Bigr) \\
    &= \min\bigl(s'_2(v),s''_2(v)\bigr).
\end{align*}

Furthermore, $\bri\in \mathcal{B}^{in}(G_x)$ if and only if $\bri|_{V_{y}}\in \mathcal{B}^{in}(G_{y})$, $\bri|_{V_{z}}\in \mathcal{B}^{in}(G_{z})$, 
and for all $v_1\in \ABI{V_{y}}$, $v_2\in \ABI{V_{z}}$, $\bri(v_1)<d_G(v_1,v_2)$ and $\bri(v_2)<d_G(v_1,v_2)$. 
By property of a tree decomposition, every path between $V_y$ and $V_z$ must go through the separator $B_x$. Hence, 
\begin{align} 
    & \forall v_1\in \ABI{V_{y}}, \forall v_2\in \ABI{V_{z}}, 
    \begin{cases}
        \bri(v_1)<d_G(v_1,v_2)\ \text{and}\\  
        \bri(v_2)<d_G(v_1,v_2)
    \end{cases}\\
    \iff 
    & \forall v_1\in \ABI{V_{y}}, \forall v_2\in \ABI{V_{z}}, 
    \begin{cases}
        \bri(v_1)<\min\limits_{v\in B_x}(d_G(v_1,v)+d_G(v,v_2))\ \text{and}\\ \bri(v_2)<\min\limits_{v\in B_x}(d_G(v_1,v)+d_G(v,v_2))\\
    \end{cases} \\
    \iff 
    & 
    \begin{cases}
        \max\limits_{v\in B_x} \bigl( \max\limits_{v_1\in \ABI{V_{y}}}(\bri(v_1)
        -d_G(v_1,v)+1)-\min\limits_{v_2\in \ABI{V_{z}}}(d_G(v_2,v)-1)\bigr)-1\leq 0\ \text{and} \\
        \max\limits_{v\in B_x} \bigl( \max\limits_{v_2\in \ABI{V_{z}}}(\bri(v_2)
        -d_G(v_2,v)+1)-\min\limits_{v_1\in \ABI{V_{y}}}(d_G(v_1,v)-1)\bigr)-1\leq 0 
    \end{cases} \label{eq:lasteq}
\end{align}

In \cref{eq:lasteq} the terms contained in the internal $\max$(s) and in the $\min$(s) are very similar to those of the signatures $s' = (s'_1, s'_2)$ and $s'' = (s''_1, s''_2)$. 
If the maximum (resp. minimum) is equal to $0$ (resp. $p$) then the equivalence still holds, which leads to the following: 
\[
    \cref{eq:lasteq} \iff \forall v\in B_x, s'_1(v)\leq s''_2(v)+1\ \text{and}\ s''_1(v)\leq s'_2(v)+1.
\]

We can deduce that the set of independent broadcasts $\bri$ having $s=(s_1,s_2)$ as a signature is exactly the set of broadcasts which verify the following:
\begin{itemize}
    \item $\bri|_{V_y}\in \mathcal{B}^{in}_y$ and $\bri|_{V_z}\in \mathcal{B}^{in}_z$
    \item let $(s'_1,s'_2)=sgn_y(\bri|_{V_y})$ and $(s''_1,s''_2)=sgn_z(\bri|_{V_z})$. Then for all $v\in B_x$:
    \[
        \begin{cases}
            s_1(v)=\max(s'_1(v),s''_1(v))\\ 
            s_2(v)=\min(s'_2(v),s''(v))\\ 
            s'_1(v)\leq s''_2(v)+1\ \text{and}\ s''_1(v) \leq s'_2(v)+1.
        \end{cases}
    \]
\end{itemize}

For all signatures $(s_1,s_2)\in A_x$, we can deduce the following recursive formula:
\begin{align*}
val_x((s_1,s_2))  = \max\Big(val_y\big((s'_1,s'_2))+val_z((s''_1,s''_2)\big)&\mid
\forall v\in B_x, \\
 s_1(v)&=\max(s'_1(v),s''_1(v)), \\
 s_2(v)&=\min(s'_2(v),s''_2(v)), \\
 s'_1(v) &\leq s''_2(v)+1,\\ 
 s''_1(v) &\leq s'_2(v)+1\Big).
\end{align*}

Regarding the complexity, assuming $val_y$ and $val_z$ have already been computed, each pair of \vals{} of $val_y$ and $val_z$ is used in the calculation of at most one \val{} of $val_x$. Therefore, we can compute it in time $O(|A_y| \cdot |A_z|)=O((p+1)^{4\tw +4})$.

\medskip

This completes the description and the proof of correctness of our algorithm. %

To conclude, since our given nice tree decomposition is of linear size and 
each step can be done in time $O((p+1)^{4\tw +4})$, this algorithm has complexity $O((p+1)^{4\tw +4})\cdot n$.
\end{proofE}

We now turn our attention to proving \cref{thm:bptw} which we recall here for the sake of readability. 

\ThmBPtw*

We once more prove that such a result holds for the weighted variant of \PBP. 
Our algorithm builds again upon a dynamic programming table computed on a nice tree decomposition of $G$ (see \cref{def:td,def:ntd}).

\mainresult{theorem}{thm:WPBPtw}
{
    The \WPBP{} problem can be solved in time $O^*((2p+2)^{O(\tw)})$ 
    on graphs of treewidth $\tw$ given a nice tree-decomposition of $G$. 
}
{progdyn}

\begin{proofE}
Recall that the \emph{natural extension} of any broadcast $\brp_x$ of ${G_x}$ is the broadcast $\brp$ defined by: 
\[
    \brp(v) = 
    \begin{cases} 
        \brp_x(v) & \text{if } v\in V_x,\\ 
        0       & \text{otherwise.}
    \end{cases}
\]

For the sake of simplicity, we do not distinguish between a broadcast of an intermediate graph $G_x$ and its natural extension to the final graph $G$. 
In the following proof, a broadcast of $G_x$, for some node $x \in V(T)$, is thus said to be a \emph{packing} if its natural extension is a broadcast packing in graph $G$. We define $\mathcal{B}(G_x)$ 
as the set of all broadcasts of $G_x$ and $\mathcal{B}^{pa}(G_x)$  
the set of all broadcast packings of $G_x$. 

\medskip

Given a node $x \in V(T)$, a \emph{signature} is a function $s : B_x \to %
\{-p-1, \ldots, p\}$.
We call $A_x$ the set of all signatures with respect to $x$. 
Note that there are for each bag at most $(2p+2)^{|B_x|} \leq (2p+2)^{tw+1}$ different signatures. %
We use such signatures on every node $x$ of $V(T)$ to provide compact representations of how a partial broadcast over $G_x$ affects the vertices of $B_x$. 

\begin{definition}[Signature of a broadcast]
\label{def:signpacking}
Let $x\in V(T)$ and $\brp\in \mathcal{B}(G_x)$. Recall that $\ABP{X}$ denotes the set of broadcasting vertices in subset of vertices $X$. Let $s\in A_x$ be defined by, for all $v\in B_x$:
\[
    s(v)=\max \Bigl(-p-1,\ \max_{\substack{u\in \ABP{V_x}}} \bigl( \brp(u) - d_G(v,u) \bigr)\Bigr).
\]
We say that $s$ is the \textit{signature of $\brp$ with respect to $x$} and define the function $sgn_x: \mathcal{B}(G_x)\rightarrow A_x$ that assigns to each broadcast its signature with respect to $x$. 
\end{definition}

Intuitively, a signature encodes, for every vertex $v\in B_x$, the distance from the nearest broadcasting neighborhood of vertices in $G_x$.
In particular, a strictly negative \val{} $s(v)=-d$ means that $v$ is in no current broadcasting vertices' broadcasting neighborhood and at distance $d$ from the nearest broadcasting neighborhood of $G_x$, and $-p-1$ is a convenient sentinel standing for \emph{no broadcasting neighborhood within distance less than or equal to $p$} in $G_x$. %
Conversely, a positive or zero \val{} $s(v)$ indicates that $v$ is
in the broadcasting range of the strongest broadcasting vertex $u$ (and then, $s(v) = s(u) - d(v, u)$). 

\medskip

Let us now describe the Dynamic Programming table. 
We associate a \val{} to each signature $s$ of a bag $x$ corresponding to the maximum \val{} of a broadcast packing with signature $s$, if it exists. 
Otherwise we set its \val{} to $-\infty$. 
In the following, we aim to compute in a bottom-up fashion the \val{} of the only signature associated to $r$, the root of the tree decomposition. 
Formally, 
for each $x \in V(T)$, we define a function $val_x:A_x\rightarrow \mathbb{N}$ as follows: 
\begin{align*}
    val_x\colon A_x & \to \mathbb{N} \\ 
   s & \mapsto \max\bigl\{\brp(V_x) \mid \brp \in \mathcal{B}^{pa}(G_x),\  sgn_x(\brp)=s \bigr\}.
\end{align*}

If $sgn_x^{-1}(s)\cap \mathcal{B}^{pa}(G_x)=\emptyset$, we set $val_x(s)=-\infty$.

\medskip

During the algorithm, an \emph{introduce node} does not decide whether the introduced vertex broadcasts or not, the decision will be taken at the corresponding \emph{forget node}. 

\medskip

Let $r$ be the root of $T$. 
By definition of a nice tree decomposition, 
$G_r=G$, and so, $\mathcal{B}^{pa}(G_r)$ is exactly the set of all the broadcast packings of $G$. 
Also, as $B_r=\emptyset$ there is only one signature $s \in A_r$. 
Therefore, $\mathcal{B}^{pa}(G_r)\subseteq sgn_r^{-1}(s)$ and $val_r(s)$ is the \val{} of the best broadcast packing of $G$. Hence, if we have an efficient way to compute the \val{} of $val_r$, then we can efficiently compute the solution to our problem. 
We now explain how to efficiently compute the \vals{} of $val_x$ for each type of node $x$. 
Take $x\in V(T)$.

\subparagraph*{Leaf node} 
If $x$ is a leaf node, then $B_x=\emptyset$ and $G_x$ is the empty graph. 
There is a single signature $s\in A_x$, and only one broadcast packing $\brp\in \mathcal{B}^{pa}(G_x)$: the empty broadcast (of \val{} 0). 
Hence $sgn_x(\brp)=s$ and $val_x(s)=0$. 
Regarding the complexity, $val_x$ can be computed in constant time.

\subparagraph*{Introduce node} 
Suppose $x$ is an introduce node, and that its unique child is $y$. 
Let $v$ be the vertex of $V$ introduced such that $B_x=B_y\cup \{v\}$. 
Take a broadcast $\brp\in \mathcal{B}(G_x)$ and $s=sgn_x(\brp)$. 
As $B_x \supseteq B_y$, we have $G_x=G_y$, 
so $\mathcal{B}(G_x)=\mathcal{B}(G_y)$ and $\mathcal{B}^{pa}(G_x)=\mathcal{B}^{pa}(G_y)$. This means that $\brp\in \mathcal{B}(G_y)$. Let $s'=sgn_y(\brp)$. As $B_x=B_y\cup \{v\}$, $s'=s|_{B_y}$. 
Furthermore, by definition of a tree decomposition, any path from 
 a vertex of $G_x$ to $v$ must go through the separator $B_y$.
Therefore: 
\begin{align*}
    s(v) & =  \max\bigl(-p-1,\max_{\substack{u\in \ABP{V_x}}}(\brp(u)-d_G(v,u))\bigr) \\ 
    & =  \max\Bigl(-p-1,\max_{\substack{u\in \ABP{V_x}}}\bigl(\max_{w\in B_y}\brp(u)-d_G(v,w)-d_G(w,u)\bigr)\Bigr) \\
    & =  \max\bigl(-p-1, \max_{w\in B_y}(s(w)-d_G(w,v))\bigr).
\end{align*} 

Let $s$ be a signature of $x$. If $s(v) = \max\bigl(-p-1, \max_{u\in B_y}(s|_{B_y}(u)-d(u,v))\bigr)$, the set of broadcast packings having $s$ as signature with respect to $x$ is exactly the set of broadcast packings having $s|_{B_y}$ as a signature with respect to $y$. Formally, $sgn_x^{-1}(s)\cap\mathcal{B}^{pa}(G_x)=sgn_y^{-1}(s|_{B_y})\cap\mathcal{B}^{pa}(G_y)$. Otherwise, no broadcast packing has $s$ as a signature.
For all signatures $s$ of $x$, we can deduce the following recursive formula:
\[val_x(s)= \begin{cases}
            val_y(s|_{B_y})& \text{if}\ s(v)=\max(-p-1, \max\limits_{u\in B_y}(s(u)-d_G(u,v))), \\
            -\infty & \text{otherwise}.
            \end{cases}\]

Regarding the complexity, assuming $val_y$ has already been computed, each \val{} of $val_x$ can be computed in time $O(|B_y|)=O(\tw)$, therefore $val_x$ can be computed in time $O(\tw \cdot |A_x|)= O(\tw \cdot (2p+2)^{\tw+1})$.
\medskip

\subparagraph*{Forget node} Suppose $x$ is a forget node and that its unique child is $y$. Let $v$ be a vertex of $V$ such that $B_y=B_x\cup \{v\}$. Take a broadcast $\brp\in \mathcal{B}(G_x)$, let $s=sgn_x(\brp)$, and $s'=sgn_y(\brp|_{V_y})$. If $\brp(v)=0$, then for all $u \in B_x$, $s(u)=s'(u)$.
Otherwise, for all $u \in B_x$, it holds that:
\begin{align*}
    s(u)&=\max\Bigl(-p-1,\max_{\substack{w\in \ABP{V_x}}}(\brp(w)-d_G(u,w))\Bigr) \\
    &=\max\Bigl(-p-1, \max_{\substack{w \in \ABP{V_y}}} (\brp(w)-d_G(u,w)),\brp(v)-d_G(u,v)\Bigr)\\
    &=\max\Bigl(s'(u),\brp(v)-d_G(u,v)\Bigr).
\end{align*}

Furthermore, $\brp\in \mathcal{B}^{pa}(G_x)$ if and only if $\brp|_{V_y}\in \mathcal{B}^{pa}(G_y)$ and either $\brp(v)=0$ or for all $w \in \ABP{V_y}$, $\brp(v) + \brp(w)-d(v,w) < 0$, i.e., $\brp(v)+s'(v)<0$. 
We can deduce that the set of broadcast packings $\brp$ having $s$ as a signature is exactly the set of broadcasts which verify one of the following properties:

\begin{itemize}
    \item $\brp|_{V_y}\in \mathcal{B}^{pa}(G_y)$, $\brp(v)=0$ and let $sgn_y(\brp|_{V_y})=s'$, where $s'|_{B_x}=s$.
    \item $\brp|_{V_y}\in \mathcal{B}^{pa}(G_y)$, $\brp(v)>0$, and let $sgn_y(\brp|_{V_y})=s'$, where $s'(v)+\brp(v)<0$ and for all $u \in B_x$, $s(u) = \max(s'(u), \brp(v) - d_G(v,u)$).
\end{itemize}
For all signatures $s\in A_x$, 
we can deduce the following recursive formula: 
\[val_x(s)=\max \begin{cases}
 val_y(s') & \mid s'|_{B_x}=s\\
 val_y(s') + l & \mid 1\leq l\leq p,\ s'(v)+l<0 \text { and for all } u \in B_x,\\
        & ~~s(u) = \max(s'(u), l - d(v,u)). 
\end{cases}\]

Regarding the complexity, assuming $val_y$ has already been computed, each \val{} of $val_y$ is used in the calculation of at most $p+1$ different \vals{} of $val_x$. Therefore, the complexity of this step is $O(p \cdot |A_y|)=O(p \cdot (2p+2)^{\tw+1})$.  

\subparagraph*{Join node} Suppose $x$ is a join node, and that its two children are $y$ and $z$. %
Let $\brp\in \mathcal{B}(G_x)$, $s'=sgn_{y}(\brp|_{V_{y}})$, $s''=sgn_{z}(\brp|_{V_{z}})$ and $s=sgn_x(\brp)$. As $V_x=V_{y}\biguplus V_{z}$ for all $v\in B_x$, %
\begin{align*}
 s(v) &= \max\bigl(-p-1, \max_{\substack{u\in \ABP{V_x}}} (\brp(u) - d_G(v,u))\bigr)\\
     &= \max\Bigl( 
            \max\bigl(-p-1, \max_{\substack{u\in \ABP{V_{y}}}} (\brp(u)-d_G(v,u))\bigr),
            \max\bigl(-p-1, \max_{\substack{u\in \ABP{V_{z}}}} (\brp(u)-d_G(v,u))\bigr)\Bigr) \\
    &= \max\bigl(s'(v),s''(v)\bigr).
\end{align*}

Furthermore, $\brp\in \mathcal{B}^{pa}(G_x)$ if and only if $\brp|_{V_{y}}\in \mathcal{B}^{pa}(G_{y})$, $\brp|_{V_{z}}\in \mathcal{B}^{pa}(G_{z})$, and for all $v_1\in \ABP{V_{y}}$, $v_2\in \ABP{V_{z}}$, $\brp(v_1)+\brp(v_2)-d(v_1,v_2)<0$. 
By property of a tree decomposition, every path between $V_{y}$ and $V_{z}$ must go through the separator $B_x$. 
Hence: 
\begin{align*}
    & \forall v_1\in \ABP{V_{y}}, \forall v_2\in \ABP{V_{z}}, \brp(v_1)+\brp(v_2)-d_G(v_1,v_2)<0 \\
    \iff 
    & \forall v_1\in \ABP{V_{y}}, \forall v_2\in \ABP{V_{z}}, \brp(v_1)+\brp(v_2)- \min_{v\in B_x}(d_G(v_1,v)+d_G(v,v_2))<0 \\
    \iff 
    & \max_{v\in B_x} \bigl( \max_{\substack{v_1\in \ABP{V_{y}}}}(\brp(v_1)
    -d_G(v_1,v))+ \max_{\substack{v_2\in \ABP{V_{z}}}}(\brp(v_2)-d_G(v_2,v)) \bigr)<0 \\
    \iff 
    & \forall v\in B_x, s'(v)+s''(v)<0.
\end{align*}

We can deduce that the set of broadcast packings $\brp$ having $s$ as a signature is exactly the set of broadcasts which verify the following:

\begin{itemize}
    \item $\brp|_{V_{y}}\in \mathcal{B}^{pa}(G_{y})$ and $\brp|_{V_{z}}\in \mathcal{B}^{pa}(G_{z})$.
    \item Let $s'=sgn_{y}(\brp|_{V_{y}})$ and $s''=sgn_{z}(\brp|_{V_{z}})$, for all $v\in B_x$, $s(v)=\max(s'(v),s''(v))$ and $s'(v)+s''(v)<0$.
\end{itemize}

For all signatures $s\in A_x$, we can deduce the following recursive formula:
\[
val_x(s) = 
\max \Bigl( val_{y}(s')+val_{z}(s'') \; 
    \mid \; 
\forall v\in B_x, \;
s(v) =\max(s'(v), s''(v)),
 \ \text{and}\ 
s'(v)+s''(v)<0 \Bigr).
 \]

Regarding the complexity, assuming $val_{y}$ and $val_{z}$ have already been computed, each pair of \vals{} of $val_{y}$ and $val_{z}$ is used in the calculation of at most one \val{} of $val_x$. Therefore, we can compute it in time $O(|A_{y}| \cdot |A_{z}|)=O((2p+2)^{2\tw+2})$.

\medskip
This completes the description and the proof of correctness of our algorithm. %
\end{proofE}

We have the following corollaries from \cref{thm:bptw,thm:bitw}. 

\apxresult{corollary}{cor:bi_fpt_vc_td}
{
    \BI{} and \BP{} can be solved in time %
    $O^*(2^{O(\td^2)})$ %
    on graphs of treedepth $\td$ given a treedepth decomposition of $G$, and in time $O^*(\vc^{O(\vc)})$ on graphs of vertex cover $\vc$ given a vertex cover of $G$.
}
{progdyn}

\begin{proofE}
    Given a graph $G$ of treewidth $\tw$, treedepth $\td$ and vertex cover $\vc$, we have $\tw \leq  \td \leq \vc+1$ 
    \cite{DBLP:books/sp/CyganFKLMPPS15,DBLP:books/daglib/0030491,Hierarchy}. 
    Furthermore, $\diam(G)\leq 2^{\td}$ \cite{DBLP:books/daglib/0030491}, and $\diam(G)\leq 2 \cdot \vc$. %
    Finally, given a treedepth decomposition of depth $\td$ or a vertex cover of size $\vc$, one can compute a tree decomposition of $G$ of width $\td$ or $\vc$ in polynomial time \cite{DBLP:books/daglib/0030491}. Therefore, we can apply \cref{thm:bitw} and \cref{thm:bptw} and see that $O^*((2^{td})^{O(td)}) = O^*(2^{O(td^2)})$ to obtain the aforementioned results.
\end{proofE}

\apxresult{corollary}{cor:BIBPtwkfpt}
{
    \BI{} and \BP{} are FPT for parameter $k+\tw(G)$.
}
{progdyn}

\begin{proofE}
    We use a standard win/win approach depending on the diameter of $G$. 
    
    \subparagraph*{Case 1: $\diam(G) > k$}  
    Consider two vertices $u,v \in G$ at distance $\diam(G)$, thus at distance at least $k$. 
    Broadcasting at \val{} $k$ from either $u$ or $v$ produces a valid (trivial) independent or broadcast packing, in polynomial time.
    
    \subparagraph*{Case 2: $\diam(G) \le k$} 
    Here, the diameter is bounded by $k$, so we can apply the FPT algorithm of \cref{thm:bptw,thm:bitw}, which are FPT for parameter $\diam(G) + \tw(G)$. Since $\diam(G) + \tw(G) \le k + \tw(G)$ the algorithm is also FPT for parameter $k + \tw(G)$. 
\end{proofE}

\section{W[1]-hardness reductions}
\label{sec:hardness}

\subsection{Proof of \cref{thm:BIFVSPWhard}}

Our reduction for \cref{thm:BIFVSPWhard} will be done through the \WBI{} problem, which we will prove hard using a reduction from the $k$-\textsc{multicolored clique}~\cite{DBLP:books/sp/CyganFKLMPPS15,DBLP:journals/jcss/Pietrzak03}. 

\paragraph*{Construction for \WBI}
Let $G_1=(I_1\cup I_2,\dots \cup I_k,E)$ be an instance of $k$-\textsc{multicolored clique}, in which %
each set $I_i$ induces an independent set of size $n$ for $i \in \bra{k}$ (we will assume without loss of generality that $n$ is even as we can always add a dummy vertex to each $I_i$). %
Recall that in such a problem one seeks a clique of size $k$ containing one vertex per set $I_i$, for $i \in \bra{k}$. 
For all $i,j\in \bra{k}$, $i<j$, we denote $E_{ij}:=E\cap \{ \{ u, v \} \mid u \in I_i, v \in I_j \}$, and $I_i:=\{v_i^1,\dots,v_i^n\}$. %

Inspired by the reduction by \authorcite{DBLP:journals/dam/KatsikarelisLP22} for the $p$-\textsc{Scattered Set} problem, %
we will construct a reduction with three types of gadgets: \emph{independent set} gadgets, that will simulate independent sets and allow to differentiate between the different vertices via small differences of the weights of the edges connecting them; 
\emph{edges gadgets}, that will simulate the edges between different vertices %
, and finally a \emph{forbid gadget}, whose sole objective is to forbid some vertices to be broadcasters in an optimal solution. Furthermore, we will make sure that all edge weights are even, which will simplify the reduction for the proof of \cref{thm:BIFVSPWhard}. %
Let us define a graph $G_2 = (V_2, E_2)$ as follows (see also \cref{fig:reduction bi weighted}). 

\begin{figure}
\centering
\begin{tikzpicture}[>=latex,scale=1.1,every node/.style={font=\small}]

\tikzstyle{bag}=[draw, blue!60!black, fill=blue!20, rounded corners, inner sep=3pt,fill opacity=0.3]
\tikzstyle{vis}=[circle, draw, inner sep=1.5pt, minimum size=6pt, fill=white, prefix after command= {\pgfextra{\tikzset{every label/.style={blue!60!black}}}}]
\tikzstyle{vforbid}=[circle, draw, inner sep=1.5pt, minimum size=6pt, fill=white, prefix after command= {\pgfextra{\tikzset{every label/.style={red!60!black}}}}]
\tikzstyle{vedge}=[circle, draw, inner sep=1.5pt, minimum size=6pt, fill=white, prefix after command= {\pgfextra{\tikzset{every label/.style={green!60!black}}}}]

\tikzstyle{forbidden}=[fill=red]

\tikzstyle{isedges}=[blue!60!black,every node/.append style={font=\scriptsize}]
\tikzstyle{edgeedges}=[green!60!black,every node/.append style={font=\scriptsize}]

\node[vis,label=right:{$v_1^1$}] (v11) {};
\node[below=2mm of v11] (v1dots) {$\vdots$};
\node[vis,below=2mm of v1dots,label=below right:{$v_1^{m_1}$}] (v1m) {};
\node[below=2mm of v1m] (v1dots2) {$\vdots$};
\node[vis,below=2mm of v1dots2,label=right:{$v_1^n$}] (v1n) {};

\node[bag, fit=(v11)(v1n)] (I1) {};

\node[above=1mm of I1] (I1label) {$I_1$};

\node[vis, forbidden, left=2cm of v1m,label=left:{$l_1$}] (l1) {};
\node[vis, forbidden, right=2cm of v1m,label=right:{$r_1$}] (r1) {};

\draw[isedges] (l1) --node[left]{$2$} (v11);
\draw[isedges] (l1) --node[above]{$2m_1$} (v1m);
\draw[isedges] (l1) --node[right]{$2n$} (v1n);

\draw[isedges] (r1) --node[right]{$2n$} (v11);
\draw[isedges] (r1) --node[above]{$2n+2-2m_1$} (v1m);
\draw[isedges] (r1) --node[left]{$2$} (v1n);

\begin{scope}[xshift=5.5cm]
\node[vis,label=right:{$v_2^1$}] (v21) {};
\node[below=2mm of v21] (v2dots) {$\vdots$};
\node[vis,below=2mm of v2dots,label=below right:{$v_2^{m_2}$}] (v2m) {};
\node[below=2mm of v2m] (v2dots2) {$\vdots$};
\node[vis,below=2mm of v2dots2,label=right:{$v_2^n$}] (v2n) {};

\node[bag, fit=(v21)(v2n)] (I2) {};

\node[above=1mm of I2] () {$I_2$};

\node[vis, forbidden, left=2cm of v2m,label=left:{$l_2$}] (l2) {};
\node[vis, forbidden, right=2cm of v2m,label=right:{$r_2$}] (r2) {};

\draw[isedges] (l2) --node[left]{$2$} (v21);
\draw[isedges] (l2) --node[above]{$2m_2$} (v2m);
\draw[isedges] (l2) --node[right]{$2n$} (v2n);

\draw[isedges] (r2) --node[right]{$2n$} (v21);
\draw[isedges] (r2) --node[above]{$2n+2-2m_2$} (v2m);
\draw[isedges] (r2) --node[left]{$2$} (v2n);

\node[right=1.2cm of r2] (troisptitspoints) {$\ldots$} ;

\end{scope}

\node[vedge, forbidden, below right=4cm of I1, label=below:{$c_{12}$}] (c12) {};
\node[vedge, forbidden, right=of c12, label=below:{$c_{13}$}] (c13) {};
\node[vedge, forbidden, right=of c13, label=below:{$c_{23}$}] (c23) {};

\foreach \dx/\dy in {-0.3/0.3, 0.3/0.3}
    \draw[edgeedges] (c12) -- ++(\dx,\dy);
    
\foreach \dx/\dy in {-0.3/0.3, 0/0.3, 0.3/0.3} {
    \draw[edgeedges] (c13) -- ++(\dx,\dy);
    \draw[edgeedges] (c23) -- ++(\dx,\dy);
}

\node[vedge, above= of c12, label=above right:{$v_e$}] (ve) {};
\draw[edgeedges] (c12) --node[right]{$n$} (ve);

\draw[edgeedges] (ve) edge[bend left] node[right]{$2b+2n+2-2m_1$} (l1);
\draw[edgeedges] (ve) edge[] node[left]{$2b+2m_1$} (r1);
\draw[edgeedges] (ve) edge[] node[right]{$2b+2n+2-2m_2$} (l2);
\draw[edgeedges] (ve) edge[bend right] node[right]{$2b+2m_2$} (r2);

\node[vforbid, left=3.5cm of c12,  label=below:{$s$}] (s) {};
\node[vforbid, minimum size=2pt, label=left:{$f_1$}, above left=9mm of s] (f1) {};
\node[vforbid, minimum size=2pt, label=left:{$f_2$}, below=2mm of f1] (f2) {};
\node[below=0.1mm of f2] (fdots) {\vdots};
\node[vforbid, minimum size=2pt, label=left:{$f_a$}, below=1mm of fdots] (fa) {};

\node[bag, red!60!black, fill=red!20, fit=(s)(f1)(fa)] (F) {};
\node[below=1mm of F] () {$F$};

\draw[red!70!black] (s) --node[above]{$\alpha$} (f1);
\draw[red!70!black] (s) -- (f2);
\draw[red!70!black] (s) -- (fa);

\draw[red!70!black] (s) --node[left]{$\alpha - 2$} (l1);
\draw[red!70!black] (s) --node[below]{$\alpha - 2$} (c12);

\foreach \dx/\dy in {0.3/0.3, 0.25/0.35, 0.35/0.25}
    \draw[red!70!black] (s) -- ++(\dx,\dy);

\node[red!70!black, below=5mm of fa] {Forbid gadget};

\foreach \forbidden in {r1, l2, r2, c13, c23} 
    \draw[red!70!black] (\forbidden) -- ($(\forbidden)!5mm!(s)$);
    
\foreach \dx/\dy in {0.3/-0.3, 0.25/-0.35, 0.35/-0.25} {
    \draw[edgeedges] (l1) -- ++(\dx,\dy);
}

\foreach \dx/\dy in {-0.3/-0.3, -0/-0.35, 0.25/-0.25} {
    \draw[edgeedges] (r1) -- ++(\dx,\dy);
}

\end{tikzpicture}
\caption{Schematic example from a graph with $k=3$ colors. For readability reasons, only the Edge gadget $e = (v_1^{m_1}, v_2^{m_2}$) is depicted and the gadget for $I_3$ is not shown. The \emph{forbid gadget} is shown in red, with $\alpha = b+n+2$. Nodes filled with red are ``forbidden'', i.e. there is an edge from $s$ to them with weight $\alpha-2$. 
The \emph{independent set gadget} of the construction is in blue, while the \emph{edge gadget} 
of the construction is in green. Weights are written on the edges.}
\label{fig:reduction bi weighted}
\end{figure}
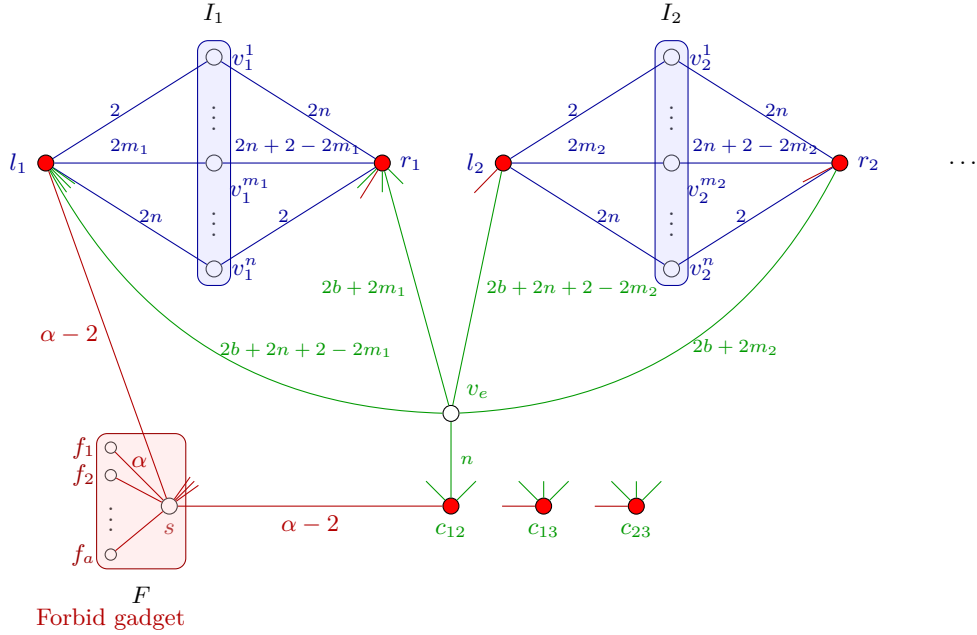

\subparagraph*{Independent set gadget} For all $i \in  \bra{k} $, we create a copy of $I_i$ (in which we keep the same vertex names for simplicity) and add two vertices $l_i,r_i$. For all $j\in \bra{n}$, we add an edge $(l_i,v_i^j)$ of weight $2j$ and an edge $(v_i^j,r_i)$ of weight $2n+2-2j$. 
Let $I'_i := I_i \cup \{l_i,r_i\}$. 
We therefore have $|I'_i| = n+2$ for each $i \in \bra{k}$ and the total number of vertices in this gadget is $k \cdot (n+2)$. 

\subparagraph*{Edges gadget} Let $b:=6n^3 \cdot k^2$. For all $i,j\in  \bra{k}$, $i<j$, we add 
a vertex $c_{ij}$.  %
For all $e=\{v_i^{m_i},v_j^{m_j}\}\in E_{ij}$, we create a vertex $v_e$, connect it to $c_{ij}$ with an edge of weight $n$, and add four edges: $(v_e,l_i)$, $(v_e,r_i)$, $(v_e,l_j)$ and $(v_e,r_j)$, of weights $2b+2n+2-2m_i$, $2b+2m_i$, $2b+2n+2-2m_j$ and $2b+2m_j$, respectively. Let $E'_{ij}:=\{v_e \mid e \in  E_{ij}\} \cup \{c_{ij}\}$.
We therefore have $\sum_{i<j}|E'_{ij}| = \sum_{i<j}(1+|E_{ij}|) = {k \choose 2} + |E| \leq {k \choose 2} + {k \choose 2} \cdot n^2$.

\subparagraph*{Forbid gadget} 
Let $a = 13k^2 \cdot n^3 \cdot b$. %
We add $a$ vertices $f_1,\dots f_a$ and a vertex $s$. 
Each $f_x$, with $x\in  \bra{a}$, is linked to $s$ with an edge of weight $\alpha := b+n+2$. 
For all $i,j\in  \bra{k}$ , with $i<j$, we connect $c_{ij}$ to $s$ with an edge of weight $ \alpha -2 $. 
For all $i\in  \bra{k}$, 
we add two edges $(l_i,s)$ and $(r_i,s)$, each of weight $\alpha -2 $. 

\medskip

Note that all weights considered in the construction are even as we assumed $n$ to be even. %
Finally, we ask whether $G_2$ admits an independent broadcast of \val{}:  
\[
    M=a \cdot \big(2\alpha -1\big) + \big(k+\tfrac{k(k-1)}{2}\big) \cdot (2b+2n+1).
\]

We first give a bound for the diameter of $G_2$. 

\apxresult{proposition}
{prop:diamG2bi}
{    
    The diameter of $G_2$ is at most $2b + 4n + 2$.
}
{hardness}

\begin{proofE}
    We first consider the eccentricity of vertex $s$. Note that its shortest distance toward all vertices except 
    $v_e \bigcup \cup_{i \in \bra{n}} I_i$ is $\alpha - 2$. By construction and since $b > n$, we have $d_{G_2}(s,v_e) = \alpha - 2 + n$. 
    Recall that $b = 5k^2n^3$%
    and that $n \geqslant 2$ since we assumed $n$ to be even. Finally, any vertex $v_i^j$ 
    from $I_i$, with $i,j \in \bra{n}$ is at distance at most $\alpha - 1 + n$ from $s$. It follows that: 
    \[
        \ecc(s) \leqslant \alpha -2 + n = b+2n.
    \]
    
    Observe that the eccentricity of any vertex $v \neq s$ is at most $\ecc(s) + d_{G_2}(v,s)$. 
    If $v = c_{ij}$ for $i,j \in \bra{k}$, $i<j$ it holds that $d_{G_2}(v,s) = \alpha-2$ and hence we have:
    \[
        \ecc(c_{ij}) \leqslant b + 2n + \alpha - 2 = 2b + 3n.
    \]
    If $v = l_i$ or $v = r_i$ for $i \in \bra{k}$ then a similar result holds since there is an edge of weight $\alpha-2$ 
    between $s$ and any such vertex. 
    If $v = v_e$, then $d_{G_2}(s,v_e) = \alpha - 2 + n$ and hence:
    \[
        \ecc(v_e) \leqslant 2 \cdot (b + 2n) = 2b + 4n.
    \]
    If $v = v_i^j$ for $i,j \in \bra{k}$, $i \neq j$, then $d(v,s)$ is at most $\alpha - 2 + n$ from $s$ and hence its eccentricity is at most $2b+4n$. 
    Finally, if $v = f_x$ for $x \in [a]$ then its eccentricity is at most $b + 2n + \alpha = 2b +4n +2$. 
    This concludes the proof of \cref{prop:diamG2bi}. 
\end{proofE}

We now prove the correctness of the reduction through the following results.

\apxresult{lemma}
{le:mcliquetoib}
{
    The graph $G_1$ has a multicolored clique $\{v_1^{m_1},\dots,v_k^{m_k}\}$ if and only if $G_2$ has an independent broadcast $\bri$ of \val{} at least $M$.
}
{hardness}

\begin{proofE}

	We first consider the forward direction. 
	Let $\bri$ be a broadcast such that for all $x\in  \bra{a}$, $\bri(f_x)=2\alpha -1=2b+2n+1$, for all $i\in  \bra{k}$, $\bri(v_i^{m_i})=2b+2n+1$, 
	for all $i,j\in \bra{k}$, $i\neq j$, %
	$\bri(v_e) = 2b+2n+1$ where $e = (v_i^{m_i},v_j^{m_j})$ is an edge of the multicolored clique $\{v_1^{m_1},\dots,v_k^{m_k}\}$,  
	and all other vertices are not broadcasting. 
	By construction, $\bri$ is independent and has \val{} exactly $M$.
	Indeed, it is not hard to see that all the distances between the broadcasters are greater or equal to $2b+2n+2$; furthermore $\bri(V_2)=a \cdot (2\alpha-1) + k \cdot (2b+2n+1) + \frac{k(k-1)}{2} \cdot (2b+2n+1)=M$.
	
	\medskip
	
    We now turn our attention to the reverse direction. 
    We proceed with a series of observations. All distances used in the following proof are taken in $G_2$. %

    \begin{claim}
	In any solution, at least one vertex $f_x$ of the forbid gadget is 
	broadcasting with \val{} at least $2\alpha -1$. 
	\end{claim}
	\begin{claimproof}
	Suppose for a contradiction that this is not the case, i.e. that $\bri(f_x)\leq 2\alpha -2$ for all $x\in  \bra{a}$.
	By \cref{prop:diamG2bi} the diameter of $G_2$ is smaller than $2b+4n+2$; 
	and the total number of vertices that are not in the forbid gadget is at most  $k \cdot(n+2+\frac{(k-1) \cdot (n^2+1)}{2})$. 
	Therefore, the total broadcast \val{} would be at most: %
	\begin{align*}
	\bri(G_2) &\leq 
	\max\limits_{x \in  \bra{a}} \bigl ( a \cdot (\bri(f_x)) + (|V_2|-a) \cdot \diam(G_2) \bigr ) \\
	&\leq 
	a \cdot \big(2\alpha -2\big) + k\big(n+2+\tfrac{(k-1)(n^2+1)}{2}\big) \cdot (2b+4n+2) \\ &
	< M,
	\end{align*}
    
	since $a > 12k^2\cdot n^3\cdot b$, which contradicts the assumption. 
	Hence, some $f_x$ has \val{} at least $2\alpha -1$ in $\bri$.
	
	\medskip
	
	Furthermore, we show that there are at least two broadcasting vertices $f_{x_1},f_{x_2}$ among the forbid gadget. 
	Indeed, if there was only one, its broadcast \val{} would be bounded by the graph diameter and 
	therefore the overall broadcast \val{} would be at most $2b+4n+2+k(n+2+\frac{(k-1)(n^2+1)}{2}) \cdot (2b+4n+2) < a < M$. %
	Therefore, there exist distinct $x_1,x_2$ such that $f_{x_1},f_{x_2}$ are broadcasting. 
	It follows that for every $x \ne x_1$, 
	\[ \bri(f_x)<d(x,x_1)= 2\alpha, \] 
	and similarly:
	\[ \bri(f_{x_1})< d(x_1,x_2)=2\alpha.\] 
	
	Thus, $\sum_{x\in \bra{a}} \bri(f_x)\leq a \cdot (2\alpha-1)$. Let us set $F:=\{f_1,\dots,f_a,s\}$ and 
	choose $x_1\in \bra{a}$ such that $\bri(f_{x_1})=2\alpha -1$. 
	Since $d(s,f_{x_1})=\alpha< \bri(f_{x_1})$, vertex $s$ cannot be broadcasting. 
	Consequently, $\bri(F)\leq a \cdot \big(2\alpha-1\big)$.
	\end{claimproof}
	
    \begin{claim}
        Every set $I'_i$ and $E'_{ij}$ contains exactly one broadcaster.
    \end{claim}
    \begin{claimproof}
	For each $i \in [k]$, every vertex of $I'_i$ lies within distance at most $2b+4n+2$ from $f_{x_1}$ (\cref{prop:diamG2bi}). %
	Hence, any broadcasting vertex of $I'_i$ has a broadcast \val{} at most $2b+4n+1$. %
	If $I'_i$ contained at least two broadcasting vertices, 
	they would be at distance at most $2n+2$ from each other, implying $\bri(v) \leq 2n+1$ for every $v \in I'_i$  
	and $\bri(I'_i)\leq n \cdot |I'_i|\leq n \cdot (n+2)$. 
	Thus, in all cases, $\bri(I'_i) < 2b +4n+1$ since $b \gg n^2$. 
	A symmetric argument holds for each $E'_{ij}$: if multiple broadcasters existed inside $E'_{ij}$, each would have broadcast \val{} bounded by $2n-1$, giving $\bri(E'_{ij})\leq |E'_{ij}| \cdot (n-1)\leq 2n^3 < 2b +4n+1$.
	
	Therefore, suppose that one edge or independent set gadget contains several or zero broadcasters. In that case and by the previous arguments, 
	the total \val{} of the broadcast would be bounded by:  
	\[
	a \cdot \big(2\alpha-1\big) 
	+ \big(k+\tfrac{k(k-1)}{2}-1\big) \cdot (2b+4n+1)
	+ 2n^3
	< M
	\]
	as $b > 5 \cdot n^3 \cdot k^2$, which is impossible. 
	\end{claimproof}
	
	We now identify which vertices can broadcast among such sets. 
	
	\begin{claim}
	    Let $i \in \bra{k}$. Then the only broadcaster of $I'_i$ belongs to $I_i$. 
	    Moreover, let $i,j \in \bra{k}$, $i < j$. Then the only broadcaster of $E'_{ij}$ is $v_e$ for some edge $e \in E_{ij}$. 
	\end{claim}
    \begin{claimproof}
	Notice first that for all $i \in \bra{k}$, $d(l_i,f_{x_1})= d(r_i,f_{x_1}) = 2\alpha -2 < \bri(f_{x_1})$, %
	so neither $l_i$ nor $r_i$ is a broadcaster. 
	Similarly, $d(c_{ij},f_{x_1})= 2 \alpha -2 < \bri(f_{x_1})$, %
	so no vertex $c_{ij}$ is a broadcaster.   
	It follows that the only broadcaster of $I'_i$ belongs to $I_i$ for $i \in \bra{k}$. 
	Similarly, the only broadcaster of $E'_{ij}$ for $i,j \in \bra{k}$, $i < j$ is $v_e$ for some edge $e \in E_{ij}$. 
    \end{claimproof}
	
	Let $v_i^{m_i}$ be the only broadcaster of $I'_i$, $i \in \bra{k}$ and 
	let $v_{e}$ be the only broadcaster of $E'_{ij}$ for some $j \in \bra{k}$ and some edge $e \in E'_{ij}$. %
	By construction, $d(v_i^{m_i},v_e)\leq 2b+2n+2$: indeed, if $e= (v_i^{m_i},v_j^{m_j})$, the further vertex from $v_e$ in $I'_i$ is 
	$v_i^{m_i}$ which is at distance exactly $2b+2n+2$ by going through $l_i$ or $r_i$. 
	It follows that $\bri(v_i^{m_i})\leq 2b+2n+1$ and $\bri(v_{e_{ij}})\leq 2b+2n+1$. 
	Therefore, the total \val{} of the broadcast is at most:
	\[
	a \cdot \big(2\alpha-1\big) + \big(k+\tfrac{k(k-1)}{2}\big) \cdot (2b+2n+1)=M.
	\]
	Hence, equality must hold:
	\[
	\bri(v_i^{m_i}) = \bri(v_{e_{ij}}) = 2b+2n+1.
	\]
	
    For all $i,j\in\bra{k}$, $i<j$, let $v_i^{m_i}$ and $v_j^{m_j}$ be the unique broadcasters of $I'_i$ and $I'_j$, respectively. 
    Moreover, let $v_{e_{ij}}$ be the only broadcaster of $E'_{ij}$, with broadcast \val{} $2b+2n+1$. 
    The following holds:
    \[
        d(v_i^{m_i},v_{e_{ij}}) > \max (\bri(v_i^{m_i}), \bri(v_{e_{ij}})) = 2b+2n+1,
    \]
    and similarly $d(v_j^{m_j},v_{e_{ij}}) > 2b+2n+1$. 
    By construction of $G_2$, the only 
    vertex of $E'_{ij}\setminus \{c_{ij}\}$ at distance more than $2b+2n+1$ from $v_i^{m_i}$ and 
    $v_j^{m_j}$ 
	is one that corresponds to an edge $(v_i^{m_i}, v_j^{m_j})$. 
	So, $e_{ij} = (v_i^{m_i}, v_j^{m_j})$, and, in particular, 
	$(v_i^{m_i}, v_j^{m_j})\in E(G_1)$.  %
    Hence, $\{v_1^{m_1},\dots,v_k^{m_k}\}$ forms a multicolored clique in $G_1$. 
    This concludes the proof of \cref{le:mcliquetoib}. 
\end{proofE}

\mainresult{theorem}{thm:WBPIVChard}
{
     \WBI{} parameterized by vertex cover is W[1]-hard even when the weights are even and polynomially bounded.
}
{hardness}

\begin{proofE}
\Cref{le:mcliquetoib} proves that $(G_1,k)$ is a yes instance of $k$-\textsc{multicolored clique} if and only if $(G_2,M)$ is a yes instance of \BI{}. Note that $S= \{l_i,r_i\}_{i \in \bra{k}}\cup \{c_{ij}\}_{i,j\in \bra{k}, i<j}\cup \{s\}$ is a vertex cover of $G_2$ of size $2k + \tfrac{k(k-1)}{2}+1$ and that all weights are polynomially bounded; furthermore, all weights are even, yielding the claimed result. 
\end{proofE}

We now prove \cref{thm:BIFVSPWhard} by transferring \cref{thm:WBPIVChard} to the unweighted case. 
Indeed, we can naturally replace edge weights by subdividing edges into paths, and add a gadget to avoid broadcasting from the internal vertices of such paths.  
This will however give a weaker result for the parameters in play. 

\paragraph*{Construction for \BI}  
Take an instance  $(G_1=(V_1,E_1),\omega:E_1\rightarrow \mathbb{N},M_1)$ of \WBI. We assume the following without loss of generality:
 \begin{itemize}
    \item All edge weights are even and polynomially bounded, 
    \item Let $\diam(G_1)$ be the (weighted) diameter of $G_1$ (so $\diam(G_1)$ is even); we may assume that all weights are bounded by $\diam(G_1)$, since edges with weight exceeding $\diam(G_1)$ can be safely removed without changing shortest path distances between vertices, 
    \item $M_1> \diam(G_1)$, as otherwise the instance is trivially positive. 
 \end{itemize}
 
 We construct an instance $(G_2=(V_2,E_2), M_2)$ of \BI{} as follows. 
 Let $n := |V_1|$, and let $a = 3  n^2 \cdot \diam(G_1)^2$. 
 We start by adding a copy of $V_1$ to $G_2$, which we still denote by $V_1$ for simplicity. 
 The remaining gadgets are similar to those from the reduction in \cref{thm:WBPIVChard} (see \cref{fig:unweighted}). 
 
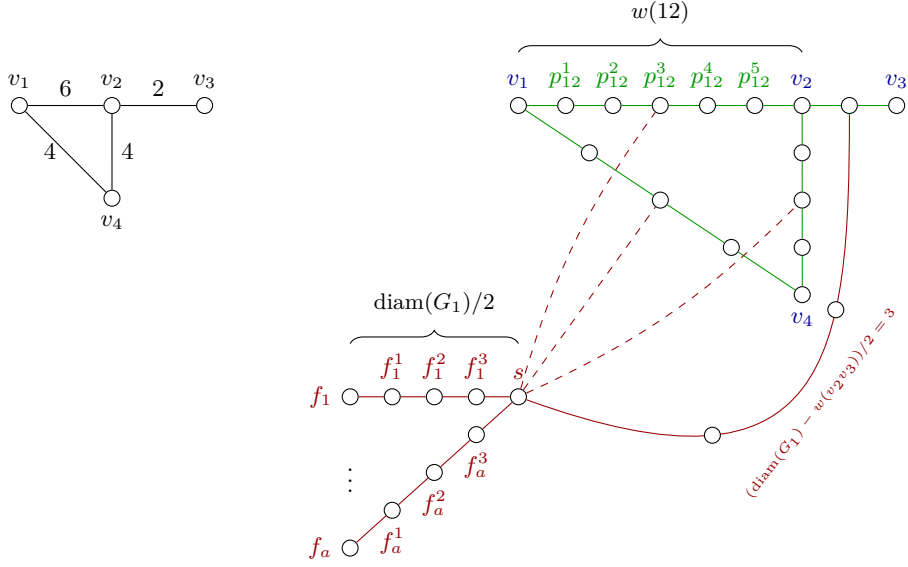
\begin{figure}
\begin{tikzpicture}[>=latex,scale=1.1,every node/.style={font=\small}]

\tikzstyle{bag}=[draw, fill=blue!10, rounded corners, inner sep=3pt,fill opacity=0.3]

\tikzstyle{vertex} =[circle, draw, inner sep=1.5pt, minimum size=6pt, fill=white]

\tikzstyle{vgadget}=[prefix after command= {\pgfextra{\tikzset{every label/.style={blue!60!black}}}}]
\tikzstyle{vedge}=[prefix after command= {\pgfextra{\tikzset{every label/.style={green!60!black}}}}]
\tikzstyle{vforbid}=[prefix after command= {\pgfextra{\tikzset{every label/.style={red!60!black}}}}]

\tikzstyle{forbidden}=[fill=red]

\tikzstyle{isedges}=[blue!60!black]
\tikzstyle{edgeedges}=[green!60!black]
\tikzstyle{forbidedges}=[red!60!black]

\node[vertex, label=above:$v_1$] (v1) {};
\node[vertex, label=above:$v_2$, right=of v1] (v2) {};
\node[vertex, label=above:$v_3$, right=of v2] (v3) {};
\node[vertex, label=below:$v_4$, below=of v2] (v4) {};

\draw[] (v1) to node[above] {$6$} (v2);
\draw[] (v2) to node[above] {$2$} (v3);
\draw[] (v1) to node[left] {$4$} (v4);
\draw[] (v2) to node[right] {$4$} (v4);

\begin{scope}[xshift=6cm]
\def\dist{0.4cm}

\node[vertex, vgadget, label=above:$v_1$] (v1) {};

\def\e{12}
\foreach \i [remember=\lastnode as \prev (initially v1)] in {1,...,5} {
    \node[vertex, vedge, right=\dist of \prev, label=above:$p_{\e}^{\i}$] (\e\i) {};
    \draw [edgeedges] (\prev) to (\e\i);
    \xdef\lastnode{\e\i}
}

\node[vertex, vgadget, label=above:$v_2$, right=\dist of \lastnode] (v2) {};
\draw [edgeedges] (\lastnode) to (v2);

\draw[decorate, decoration={brace, amplitude=5pt, raise=20pt}]
    ($(v1)!0!(v2)$) -- ($(v1)!1!(v2)$)
    node[midway, yshift=35pt, font=\footnotesize] {$w(12)$};

\foreach \i [remember=\lastnode as \prev (initially v2)] in {1,...,3} {
    \node[vertex, vedge, below=\dist of \prev] (24\i) {};
    \draw [edgeedges] (\prev) to (24\i);
    \xdef\lastnode{24\i}
}
\node[vertex, vgadget, label=below:$v_4$, below=\dist of \lastnode] (v4) {};
\draw [edgeedges] (\lastnode) to (v4);

\foreach \i/\name [remember=\name as \prev (initially v1)] in {0.25/141, 0.5/142, 0.75/143} {
    \node[vertex, vedge] (\name) at ($(v1)!\i!(v4)$) {};
    \draw [edgeedges] (\prev) to (\name);
    \xdef\lastnode{\name}
}
\draw [edgeedges] (\lastnode) to (v4);

\foreach \i [remember=\lastnode as \prev (initially v2)] in {1,...,1} {
        \node[vertex, vedge, right=\dist of \prev] (23\i) {};
        \draw [edgeedges] (\prev) to (23\i);
        \xdef\lastnode{23\i}
    }

\node[vertex, vgadget, label=above:$v_3$, right=\dist of \lastnode] (v3) {};
\draw [edgeedges] (\lastnode) to (v3);

\begin{scope}[yshift=-3.5cm]
    \node[vertex, vforbid, label=above:$s$] (s) {};
    \node[vertex, vforbid, label=left:{$f_1$}, left=20mm of s] (f1) {};
    \node[below=5mm of f1] (fdots) {\vdots};
    \node[vertex, vforbid, label=left:{$f_a$}, below=5mm of fdots] (fa) {};

    \foreach \i/\name [remember=\name as \lastnode (initially f1)] in {0.25/f_1^1, 0.5/f_1^2, 0.75/f_1^3} {
        \node[vertex, vforbid, label=above:{$\name$}] (\name) at ($(f1)!\i!(s)$) {};
        \draw [forbidedges] (\lastnode) to (\name);
        \xdef\lastnode{\name}
    }

    \draw[decorate, decoration={brace, amplitude=5pt, raise=20pt}]
    ($(f1)!0!(s)$) -- ($(f1)!1!(s)$)
    node[midway, yshift=35pt,font=\footnotesize] {$\diam(G_1)/2$};

    \draw [forbidedges] (\lastnode) to (s);
        \foreach \i/\name [remember=\name as \lastnode (initially fa)] in {0.25/f_a^1, 0.5/f_a^2, 0.75/f_a^3} {
        \node[vertex, vforbid, label=below:{$\name$}] (\name) at ($(fa)!\i!(s)$) {};
        \draw [forbidedges] (\lastnode) to (\name);
        \xdef\lastnode{\name}
    }
    \draw [forbidedges] (\lastnode) to (s);

    \draw [forbidedges] (s) edge[dashed, bend left=10] (123);
    \draw [forbidedges] (s) edge[dashed] (142);
    \draw [forbidedges] (s) edge[bend right=10, dashed] (242);

    \draw[
    forbidedges,
    postaction={
        decorate,
        decoration={
            markings,
            mark=at position 0.33 with {\node[vertex,draw=black] {};},
            mark=at position 0.66  with {\node[vertex,draw=black] {};},
        }
    }
   ] 
(s) .. controls +(-20:4cm) and +(-90:3cm) .. (231)
node[midway, below=1mm, sloped, font=\tiny] {$(\diam(G_1)-w(v_2 v_3))/2 = 3$}; 

\end{scope}

\end{scope}
\end{tikzpicture}
\caption{\label{fig:unweighted}On the left: an edge-weighted graph $G_1$, part of the instance of \WBI, of diameter 8. On the right: the graph $G_2$, part of the instance of \BI, obtained from $G_1$ after applying the proposed reduction}
\end{figure}%

\subparagraph*{Forbid Gadget} We add a vertex $s$ and vertices $f_1,\dots, f_a$.
For each $x\in \bra{a}$, add a path of length $\frac{diam(G_1)}{2}$ between $f_x$ and $s$ and denote this path by $(f_x,f_x^1,\dots f_x^{\frac{\diam(G_1)}{2}-1},s)$. 

Let $F_{x}:=\{f_x,f_x^1,\dots f_x^{\frac{\diam(G_1)}{2}-1}\}$ and $F := \{s\}\cup\bigcup_{x=1}^a F_x$ (that is $F$ contains all vertices of the forbid gadget).
The aim of this gadget is to forbid some vertices of the graph to broadcast. 
To do so, the vertices in the forbid gadget will broadcast with sufficiently large \val{} to include the forbidden vertices in their broadcast neighborhood.
 
\subparagraph*{Edge gadget}  
For every edge $e=(u,v)\in E_1$ with weight $\omega(e)$, add a path of length $\omega(e)$ between the corresponding copies $u$ and $v$ in $V_1$:
\[(u,p_e^1,\dots, p_e^{\omega(e)-1},v).\] 

Also connect $s$ to the middle vertex $p_e^{\frac{\omega(e)}{2}}$ (recall that all weights are even in $G_1$) with a path of length $\dfrac{(\diam(G_1)-{\omega(e))}}{2}$, denoted as follows: 
\[(s,q_e^1,\dots,q_e^{\frac{\diam(G_1)-\omega(e)}{2}-1},p_e^{\frac{\omega(e)}{2}}).\]
Let $V_e$ be the set of all vertices introduced in the edge gadget for edge $e$. 
Because all weights and $\diam(G_1)$ are polynomially bounded, the construction is polynomial in size. Finally, set: 
\[
    M_2:=a \cdot (\diam(G_1)-1) + M_1.
\] 
 
Let $d_1$ and $d_2$ denote the (weighted) distance functions in $G_1$ and $G_2$ respectively. %
By construction, for all $u,v \in V_1$, we have $d_2(u,v) = d_1(u,v)$, and the diameter of $G_2$ is exactly equal to $\diam(G_1)$.

\apxresult{lemma}
{le:wibtoib}
{
    The graph $G_1$ has an independent broadcast $\bri_1$ of \val{} at least $M_1$ if and only if 
    the graph $G_2$ has an independent broadcast $\bri_2$ of \val{} at least $M_2$.
}
{hardness}

\begin{proofE}
We first consider the forward direction. 
Define $\bri_2$ by: 
\[
\bri_2(v) =
\begin{cases}
\bri_1(v), & v\in V_1,\\
\diam(G_1)-1, & v = f_x \text{ for some } x\in \bra{a},\\
0, & \text{otherwise.}
\end{cases}
\] 

We first claim that $\bri_2$ is an independent broadcast. 
For any distinct $x,x' \in \bra{a}$ and any $v \in V_1$, we have
\(d_2(f_x,v) = d_2(f_x,f_{x'}) = \diam(G_1).\) %
Therefore, given $x \in \bra{a}$, no vertex from $(\{ f_1, \dots, f_a \} \setminus \{f_x\}) \cup V_1 $ hears $f_x$. %
Since $\bri_1$ is an independent broadcast of $G_1$ and since distances inside $V_1$ are not changed, for all $v,v'\in V_1$ such that $\bri_1(v) >0$ and $\bri_1(v') >0$ it holds that $\bri_2(v)=\bri_1(v)<d_1(v,v')=d_2(v,v')$. 
Hence $\bri_2$ is indeed an independent broadcast, of \val: 
\[ \bri_2(V_2) = \bri_1(V_1) + \sum_{x=1}^a (\diam(G_1)-1) \ge M_1 + a \cdot (\diam(G_1)-1) = M_2 \]
as required.

\medskip

We now turn our attention to the reverse direction. 
We first bound the contribution of the forbid gadget, and show that $\bri_2(F) \le a \cdot (\diam(G_1)-1)$ and that $F$ 
contains a broadcaster $b$ with $\bri_2(b) = \diam(G_1)-1$. 

To that aim, we first prove that there exists $x \in \bra{a}$ such that $\bri_2(F_x)\geq \diam(G_1)-1$. %
Assume for contradiction that for all $x\in \bra{a}$, $\bri_2(F_x)\leq \diam(G_1)-2$. 
Then, since $|V_2\setminus \bigcup_{x\in \bra{a}}F_x|\leq 2 \cdot \diam(G_1) \cdot n^2$, %
and since $a \gg n^2 \cdot \diam(G_1)^2$ by choice of $a$, we have: 
\[ \bri_2(V_2) \le a \cdot (\diam(G_1)-2)+ 2 \cdot \diam(G_1)^2 \cdot n^2<M_2,\]
because $a>2n^2 \cdot \diam(G_1)^2$, a contradiction. 
Hence, there exists $x \in \bra{a}$ with $\bri_2(F_x)\geq \diam(G_1)-1$.  

Assume such an $F_x$ contains multiple broadcasters $b_1,\dots,b_\ell$, ordered along the path formed by $F_x$. Then, by independence of the solution, their contribution is limited by the pairwise distances along the path:
\begin{align*}
\bri_2(F_x)=\sum_{i=1}^{\ell}\bri_2(b_i) & 
\leq \bigl(d_2(b_1,b_2) - 1 \bigr) + \sum_{i=2}^{\ell} \bigl(d_2(b_{i-1},b_{i}) - 1 \bigr) \\
& \leq \frac{\diam(G_1)}{2}-1+\frac{\diam(G_1)}{2}-(\ell-1) \\
& <\diam(G_1)-1,
\end{align*} 
leading to a contradiction. 
Thus, $F_x$ contains exactly one broadcaster $b$, and $\bri_2(b)=\bri_2(F_x)\geq \diam(G_1)-1$. 
As the diameter of $G_2$ is equal to $\diam(G_1)$, $\bri_2(b)\leq \diam(G_1)$. 
But if $\bri_2(b)=\diam(G_1)$, then $b$ is the only broadcaster of the graph, 
which is impossible as $\diam(G_1)< M_2$. 
So %
$\bri_2(b)=\diam(G_1)-1$ and for all $x'\in \bra{a}$, $\bri_2(F_{x'})\leq \diam(G_1)-1$ holds.
Furthermore, $s$ is not a broadcaster because $d(s,b) \leq \diam(G_1)/2$.
Therefore, 
\[
    \bri_2(F)
    = \sum_{x=1}^a \bri_2(F_x)
    \leq a \cdot (\diam(G_1)-1).
\]

Next, we show that no vertex of any edge gadget $V_e$ can be a broadcaster. 
Let $v_e \in V_e$ for some edge $e = uv$ in $E_1$. 
We show that in all cases, $d_2(v_e,b) < \diam(G_1)$ %
Indeed, if $v_e$ lies on the path from $s$ to the middle vertex $p_e^{\omega(e)/2}$, i.e. $v_e \in \{q_e^1,\dots,q_e^{\frac{\diam(G_1)-\omega(e)}{2}-1}\}$, then: 
\begin{align*}
    d_2(v_e,b) & \leq d_2(v_e,s)+d_2(s,b) \\
    & \leq \frac{\diam(G_1)}{2}-1+\frac{\diam(G_1)}{2} \\
    & < \diam(G_1).
\end{align*}
Otherwise, $v_e$ lies on the path from $u$ to $v$, and:
\begin{align*}
    d_2(v_e,b)
    & \leq d_2(v_e,p_e^{\omega(e)/2})+d_2(p_e^{\omega(e)/2},s)+d_2(s,b) \\
    & \leq \frac{\omega(e)}{2}-1+ \frac{\diam(G_1)-\omega(e)}{2} +\frac{\diam(G_1)}{2} \\
    & < \diam(G_1).
\end{align*} 
Therefore, there is no broadcaster in $V_e$. Thus, all broadcasters outside $F$ are in $V_1$ and: %
\[
\bri_2(V_2) = \bri_2(F) + \bri_2(V_1) \le a \cdot (\diam(G_1)-1) + \bri_2(V_1).
\]

Since $\bri_2(V_2) \geq M_2=a \cdot (\diam(G_1)-1) + M_1$, we deduce 
\(\bri_2(V_1)\geq M_1.\)

\medskip

Define $\bri_1:=\bri_2|_{V_1}$. 
It is a valid independent broadcast of $G_1$: by construction $\bri_1(v) \le \diam(G_1)$ for any $v \in V_1$, so it is a broadcast. Moreover, since $d_1:=d_2|_{V_1\times V_1}$ and $\bri_2$ is independent, $\bri_1$ also is. %
Finally, $\bri_1(V_1)= \bri_2(V_1)\geq M_1$. 
\end{proofE}

\ThmBIFVSPWhard*

\begin{proof}
\Cref{le:wibtoib} proves that $(G_1, M_1)$ is a yes instance of \WBI{} if and only if $(G_2,M_2)$ is a yes instance of \BI{}. 
Let $X\subseteq V_1$ be a vertex cover of $G_1$, such that $\vc(G_1) := |X|$. 
Removing $X\cup\{s\}$ from $G_2$ leaves a disjoint union of trees, each of pathwidth at most~$3$. %
Therefore, 
\[
\fvs(G_2) \leq \vc(G_1)+1 \quad \text{ and } \quad \pw(G_2)\leq \vc(G_1)+4.
\]
This concludes the proof of \cref{thm:BIFVSPWhard}. 
\end{proof}

\subsection{Proof of \cref{thm:WBPVChard}}

We finally turn our attention to the \BP{} problem and provide the reduction for \cref{thm:WBPVChard}. We once again present a reduction from $k$-\textsc{multicolored clique}~\cite{DBLP:books/sp/CyganFKLMPPS15,DBLP:journals/jcss/Pietrzak03}.

\paragraph*{Construction for \WBP}  
Let $G_1=(I_1\cup I_2,\dots \cup I_k,E)$ be an instance of $k$-\textsc{multicolored clique}, in which %
each set $I_i$ induces an independent set of size $n$ for $i \in \bra{k}$ (we will assume without loss of generality that $n$ is even as we can always add a dummy vertex to each $I_i$). %
Recall that in such a problem one seeks a clique of size $k$ containing one vertex per set $I_i$, for $i \in \bra{k}$. 
For each pair $i,j \in \bra{k}$, $i<j$, let $E_{ij}:=E\cap (I_i\times I_j)$ and  $I_i:=\{v_i^1,\dots,v_i^n\}$. %

\medskip

The construction resembles the one in the proof of \cref{thm:WBPIVChard}, yet with some key differences. We notably replace the previous forbid gadget with a \emph{bounding gadget} that limits the broadcast \vals{} of the other vertices of the graph. Let us define a graph $G_2 = (V_2, E_2)$ as follows. 

\subparagraph*{Independent set gadget} For each color class $i \in  \bra{k}$, create a copy of $I_i$ (we keep the same vertex names for simplicity), and two vertices $l_i,r_i$. 
For each $j\in \bra{n}$, add edge $(l_i,v_i^j)$ of weight $j+1$, and edge $(v_i^j,r_i)$ of weight $n+2-j$. 
Let $I'_i := I_i \cup \{l_i,r_i\}$.
We therefore have $|I'_i| = n+2$ for each $i \in \bra{k}$ and the total number of vertices in this gadget is $k \cdot (n+2)$. 

\subparagraph*{Edges gadget} Let $b:=5k^2 \cdot n^3$. 
For each pair $i<j$, add a vertex $c_{ij}$. 
For each edge $e=(v_i^{m_i},v_j^{m_j}) \in E_{ij}$, add a vertex $v_e$, connect it to $c_{ij}$ by an edge of weight $\dfrac{n}{2}$, and add the four edges 
$(v_e,l_i),\ (v_e,r_i),\ (v_e,l_j),\ (v_e,r_j)$
of weights $2b+2n-m_i$, $2b+n+m_i-1$, $2b+2n-m_j$, and $2b+n+m_j-1$, respectively. 
Let $E'_{ij}:=\{v_e \mid e \in  E_{ij}\} \cup \{c_{ij}\}$.
The total number of vertices in this gadget is thus $\sum_{i<j}|E'_{ij}| = \sum_{i<j}(1+|E_{ij}|) = {k \choose 2} + |E| \leq {k \choose 2} + {k \choose 2} \cdot n^2$. %

\subparagraph*{Bounding gadget} Let $a = 11 k^2 \cdot n^3 \cdot b$. %
We add $a$ vertices $f_1,\dots f_a$ and a vertex $s$.
For each $x \in \bra{a}$, connect $f_x$ to $s$ with an edge of weight $2$. 
Connect $s$ 
to each $l_i$ and $r_i$ with edges of weight $b+n$, 
to each $v_i^j$ with an edge of weight $b+n+1$, 
to each $c_{ij}$ with an edge of weight $b+n-1$,
and to each $v_e$ with an edge of weight $b+n+1$. 
Let $F:=\{f_x \mid x\in\bra{a}\}\cup\{s\}$. Intuitively, this gadget will force other vertices of the graph to have bounded broadcasting \vals. Note that $s$ is universal in $G_2$.  

\begin{figure}
\centering
\begin{tikzpicture}[>=latex,scale=1.1,every node/.style={font=\small}]

\tikzstyle{bag}=[draw, blue!60!black, fill=blue!20, rounded corners, inner sep=3pt,fill opacity=0.3]
\tikzstyle{vis}=[circle, draw, inner sep=1.5pt, minimum size=6pt, fill=white, prefix after command= {\pgfextra{\tikzset{every label/.style={blue!60!black}}}}]
\tikzstyle{vbounding}=[circle, draw, inner sep=1.5pt, minimum size=6pt, fill=white, prefix after command= {\pgfextra{\tikzset{every label/.style={red!60!black}}}}]
\tikzstyle{vedge}=[circle, draw, inner sep=1.5pt, minimum size=6pt, fill=white, prefix after command= {\pgfextra{\tikzset{every label/.style={green!60!black}}}}]

\tikzstyle{bounding}=[]

\tikzstyle{isedges}=[blue!60!black,every node/.append style={font=\scriptsize}]
\tikzstyle{edgeedges}=[green!60!black,every node/.append style={font=\scriptsize}]
\tikzstyle{boundedges}=[red!70!black,every node/.append style={font=\scriptsize}]

\node[vis,label=right:{$v_1^1$}] (v11) {};
\node[below=2mm of v11] (v1dots) {$\vdots$};
\node[vis,below=2mm of v1dots,label=below right:{$v_1^{m_1}$}] (v1m) {};
\node[below=2mm of v1m] (v1dots2) {$\vdots$};
\node[vis,below=2mm of v1dots2,label=right:{$v_1^n$}] (v1n) {};

\node[bag, fit=(v11)(v1n)] (I1) {};

\node[above=1mm of I1] (I1label) {$I_1$};

\node[vis, bounding, left=2cm of v1m,label=left:{$l_1$}] (l1) {};
\node[vis, bounding, right=2cm of v1m,label=right:{$r_1$}] (r1) {};

\draw[isedges] (l1) --node[left]{$2$} (v11);
\draw[isedges] (l1) --node[above]{$m_1+1$} (v1m);
\draw[isedges] (l1) --node[right]{$n+1$} (v1n);

\draw[isedges] (r1) --node[right]{$n+1$} (v11);
\draw[isedges] (r1) --node[above]{$n+2-m_1$} (v1m);
\draw[isedges] (r1) --node[left]{$2$} (v1n);

\begin{scope}[xshift=5.5cm]
\node[vis,label=right:{$v_2^1$}] (v21) {};
\node[below=2mm of v21] (v2dots) {$\vdots$};
\node[vis,below=2mm of v2dots,label=below right:{$v_2^{m_2}$}] (v2m) {};
\node[below=2mm of v2m] (v2dots2) {$\vdots$};
\node[vis,below=2mm of v2dots2,label=right:{$v_2^n$}] (v2n) {};

\node[bag, fit=(v21)(v2n)] (I2) {};

\node[above=1mm of I2] () {$I_2$};

\node[vis, bounding, left=2cm of v2m,label=left:{$l_2$}] (l2) {};
\node[vis, bounding, right=2cm of v2m,label=right:{$r_2$}] (r2) {};

\draw[isedges] (l2) --node[left]{$2$} (v21);
\draw[isedges] (l2) --node[above]{$m_2+1$} (v2m);
\draw[isedges] (l2) --node[right]{$n+1$} (v2n);

\draw[isedges] (r2) --node[right]{$n+1$} (v21);
\draw[isedges] (r2) --node[above]{$n+2-m_2$} (v2m);
\draw[isedges] (r2) --node[left]{$2$} (v2n);

\node[right=1.2cm of r2] (troisptitspoints) {$\ldots$} ;

\end{scope}

\node[vedge, bounding, below right=4cm of I1, label=below:{$c_{12}$}] (c12) {};
\node[vedge, bounding, right=of c12, label=below:{$c_{13}$}] (c13) {};
\node[vedge, bounding, right=of c13, label=below:{$c_{23}$}] (c23) {};

\foreach \dx/\dy in {-0.3/0.3, 0.3/0.3}
    \draw[edgeedges] (c12) -- ++(\dx,\dy);
    
\foreach \dx/\dy in {-0.3/0.3, 0/0.3, 0.3/0.3} {
    \draw[edgeedges] (c13) -- ++(\dx,\dy);
    \draw[edgeedges] (c23) -- ++(\dx,\dy);
}

\node[vedge, above= of c12, label=above right:{$v_e$}] (ve) {};
\draw[edgeedges] (c12) --node[right]{$n/2$} (ve);

\draw[edgeedges] (ve) edge[bend left] node[right]{$2b+2n-m_1$} (l1);
\draw[edgeedges] (ve) edge[] node[left]{$2b+n+m_i-1$} (r1);
\draw[edgeedges] (ve) edge[] node[right]{$2b+2n-m_2$} (l2);
\draw[edgeedges] (ve) edge[bend right] node[right]{$2b+n+m_2-1$} (r2);

\node[vbounding, left=3.5cm of c12,  label=below:{$s$}] (s) {};
\node[vbounding, minimum size=2pt, label=left:{$f_1$}, above left=9mm of s] (f1) {};
\node[vbounding, minimum size=2pt, label=left:{$f_2$}, below=2mm of f1] (f2) {};
\node[below=0.1mm of f2] (fdots) {\vdots};
\node[vbounding, minimum size=2pt, label=left:{$f_a$}, below=1mm of fdots] (fa) {};

\node[bag, red!60!black, fill=red!20, fit=(s)(f1)(fa)] (F) {};
\node[below=1mm of F] () {$F$};

\draw[boundedges] (s) -- (f1);
\draw[boundedges] (s) -- (f2);
\draw[boundedges] (s) --node[below]{2} (fa);

\draw[boundedges] (s) --node[left]{$b + n$} (l1);
\draw[boundedges] (s) --node[below]{$b + n -1$} (c12);
\draw[boundedges] (s) --node[above]{$b + n +1$} (ve);
\draw[boundedges] (s) --node[right]{$b + n +1$} (v1n);

\foreach \dx/\dy in {0.3/0.3, 0.4/0.2, 0.2/0.4, 0.45/0.08, 0/0.5}
    \draw[boundedges] (s) -- ++(\dx,\dy);

\node[red!70!black, below=5mm of fa] {Bounding gadget};

\foreach \bounding in {r1, l2, r2, c13, c23, v11, v1m, v21, v2m, v2n} %
    \draw[boundedges] (\bounding) -- ($(\bounding)!5mm!(s)$);
    
\foreach \dx/\dy in {0.3/-0.3, 0.25/-0.35, 0.35/-0.25} {
    \draw[edgeedges] (l1) -- ++(\dx,\dy);
}

\foreach \dx/\dy in {-0.3/-0.3, -0/-0.35, 0.25/-0.25} {
    \draw[edgeedges] (r1) -- ++(\dx,\dy);
}

\end{tikzpicture}
\caption{Schematic example from a graph with $k=3$ colors. 
For readability reasons, only the Edge gadget $e = (v_1^{m_1}, v_2^{m_2}$) is depicted and the gadget for $I_3$ is not shown. 
The \emph{bounding gadget} is shown in red. 
The \emph{independent set gadget} of the construction is in blue, while the \emph{edge gadget} 
of the construction is in green. Weights are written on the edges.
\label{fig:bp weighted}
}
\end{figure}
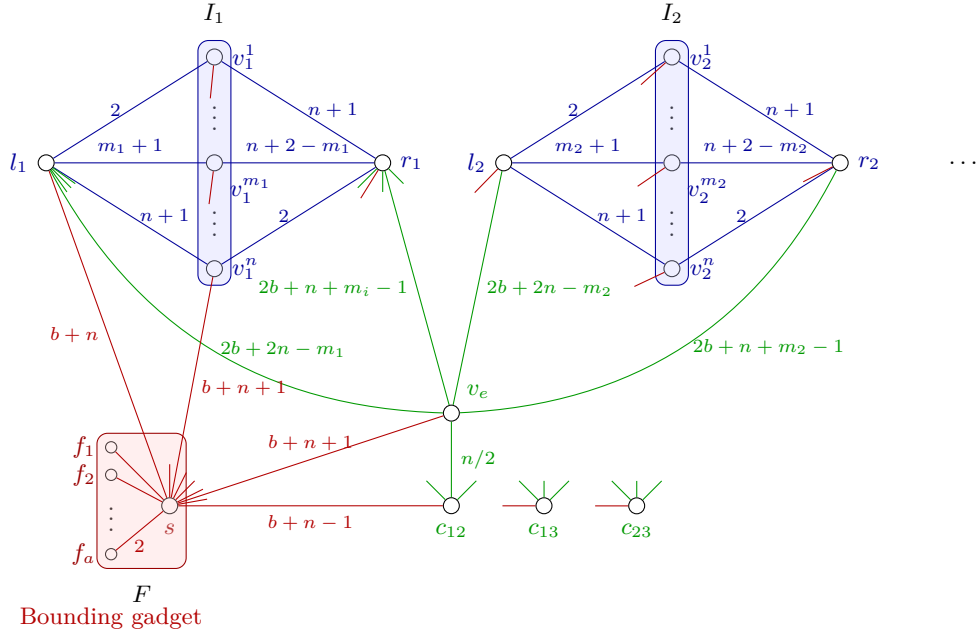
    
Finally, we ask whether $G_2$ admits a broadcast packing of total \val:
\[
    M = a +1 + (k+\frac{k(k-1)}{2}) \cdot (b+n).
\]

Note that all weights used are polynomial. The construction is depicted in the appendix (see \cref{fig:bp weighted}).  

\apxresult{proposition}
{prop:diamG2bp}
{    
    The diameter of $G_2$ is at most $2b + 2n + 2$.
}
{hardness}

\begin{proofE}
    We first consider the eccentricity of vertex $s$. Since it is universal, its eccentricity corresponds to the 
    maximum weight of any edge incident to it, that is:
    \[
        \ecc(s) \leqslant b + n + 1.
    \]
    
    Observe that the eccentricity of any vertex $v \neq s$ is at most $\ecc(s) + d_{G_2}(v,s)$. 
    If $v = c_{ij}$ for $i,j \in \bra{k}$, $i<j$ it holds that $d_{G_2}(v,s) = b+n-1$ and hence we have:
    \[
        \ecc(c_{ij}) \leqslant 2b + 2n.
    \]
    If $v = l_i$ or $v = r_i$ for $i \in \bra{k}$ it holds that $\ecc(v) \leqslant 2b + 2n + 1$. 
    If $v = v_e$ or $v = v_i^j$ for $i,j \in \bra{k}$, $i < j$, 
    then $d_{G_2}(s,v_e) = b + n + 1$ and hence $\ecc(v_e) \leqslant 2b + 2n + 2$. 
    Finally, if $v = f_x$ for $x \in [a]$ then $\ecc(v)$ is at most $b+n+3$. 
    This concludes the proof of \cref{prop:diamG2bp}. 
\end{proofE}

\apxresult{lemma}
{le:bptomclique}
{ 
    The graph $G_1$ has a multicolored clique $\{v_1^{m_1},\dots,v_k^{m_k}\}$ if and only if 
    the graph $G_2$ has a broadcast packing of \val{} at least $M$.
}
{hardness}

\begin{proofE}
We first prove the forward direction. Let $\brp$ be a broadcast such that $\brp(f_1)=2$, $\brp(f_x)=1$ for all $x \ge 2$, $\brp(v_i^{m_i})=b+n$ for each $i\in \bra{k}$, $\brp(v_{(v_i^{m_i},v_j^{m_j})})=b+n$ for each pair $i<j$, and $\brp(u) = 0$ otherwise. 
By construction, it is a packing and it has \val{} exactly $M$. 

\medskip

We now turn our attention to the reverse direction. We proceed with a series of observations. Note that all the distances used in the following proof are in $G_2$. %

\begin{claim}
    Any solution must have at least one vertex $f_x$ as a broadcaster. 
\end{claim}
\begin{claimproof}
Assume for a contradiction that this is not the case. This means that there are at most $|V_2| - a$ broadcasters. Moreover, each broadcasting vertex has a \val{} bounded by its eccentricity, which is at most $\diam(G_2) \leq 2b + 2n + 2$. 
Hence,
\begin{align*}
    \brp(V_2) & \leq (|V_2|-a) \cdot \diam(G_2)\\
    & \leq \big( k \cdot (n+2+\frac{(k-1)(n^2+1)}{2}) + 1 \big) \cdot (2b+2n+2)\\
    & < a < M
\end{align*}

since $a > 10k^2\cdot n^3\cdot b$, contradicting the assumption. Hence, at least one vertex $f_x$ is a broadcaster.
\end{claimproof}

\begin{claim}
    There are least two vertices $f_{x_1},f_{x_2}$ that are broadcasters among the bounding gadget. 
\end{claim}
\begin{claimproof}
Assume for a contradiction that there is only one such vertex. Then its contribution is bounded by the graph diameter and therefore the overall broadcast \val{} is at most $2b+2n+2 + k(n+2+\frac{(k-1)(n^2+1)+1}{2}) \cdot (2b+2n+2) < M$.
\end{claimproof}

Therefore there exist $x_1,x_2\in \bra{a}$, $x_1\neq x_2$, such that $f_{x_1},f_{x_2}$ are broadcasters. 

\begin{claim}
    \label{cl:overall}
    The overall contribution of $F$ to $\brp$ is bounded by $a+1$.
\end{claim}
\begin{claimproof}
    For all $x\in \bra{a}\setminus\{x_1\}$, 
\[\brp(f_x)<d(f_x,f_{x_1})-\brp(f_{x_1})\leq 3\]

and similarly: 
\[\brp(f_{x_1})<d(f_{x_1},f_{x_2})-\brp(f_{x_2})\leq3.\]

So, for all $x\in\bra{a}$, $\brp(f_x)\leq 2$. But, suppose there is $x_1\in \bra{a}$ such that $\brp(f_{x_1})=2$. Then, for all $x\in \bra{a}\setminus \{x_1\}$,
\[\brp(f_x)<d(f_x,f_{x_1})-\brp(f_{x_1})=2.\]

So, for all $x\in\bra{a}\setminus \{x_1\}$, $\brp(f_x)\leq 1$.

\medskip

Finally, let $x_1\in \bra{a}$ such that $f_{x_1}$ is a broadcaster. Then $\brp(s)< d(f_{x_1},s)-\brp(f_{x_1})\leq 1$, and therefore, $s$ is not a broadcaster: $\brp(s)=0$. 
We can deduce that: 
\[
    \brp(F) = \brp(s)+\sum_{x\in\bra{a}}\brp(f_x)\leq a+1.
\]
This concludes the proof of \cref{cl:overall}. 
\end{claimproof}

\begin{claim}
    \label{cl:exactlyone}
    Every set $I'_i$ and $E'_{ij}$ contains exactly one broadcaster.
\end{claim}
\begin{claimproof}
Let $x_1\in\bra{a}$ be such that $\brp(f_{x_1})\geq 1$. 
Note that every vertex of $V_2\setminus F$ lies within distance at most $b+n+3$ from $f_{x_1}$ through $s$. %
Hence, any broadcaster in $V_2\setminus F$ has a broadcast \val{} strictly less than $b+n+3 - \brp(f_{x_1}) \leq b+n+2$. 
Moreover, within each independent set gadget $I'_i$, two broadcasters would be at distance at most $n+3$ from each other, 
implying that for all $v \in I'_i$, $\brp(v)\leq n+2$, and thus forcing $\brp(I'_i) \leq |I'_i| \cdot (n+2) = (n+2)^2$. 

An analogous argument holds for all edge gadgets $E'_{ij}$:  if multiple broadcasters are inside $E'_{ij}$, each of their broadcast \vals{} are bounded by $n-1$, giving $\brp(E'_{ij}) \leq |E'_{ij}| \cdot (n-1) \leq n^3$ %

\medskip

Therefore, suppose that one edge or independent set gadget contains several or zero broadcasters. In that case and by the previous arguments, %
the overall broadcast \val{} would be at most: 
\[
	a+1 
	+ \big(k+\tfrac{k(k-1)}{2}-1\big)(b+n+1)
	+ 2n^3
	< M,
\]
which is impossible.
\end{claimproof}

\begin{claim}
    \label{cl:bound}
    Every set $I'_i$ and $E'_{ij}$ has broadcast \val{} exactly $b+n$.
\end{claim}
\begin{claimproof}
By \cref{cl:exactlyone}, every such set contains exactly one broadcaster. We first turn our attention to independent set gadgets. 
Assume for a contradiction that a vertex $v$ in set $I'_i$ broadcast at \val{} $\brp(v)=b+n+1$. 
Since $v$ is at distance at most $2b+2n+2$ from all other vertices, %
any other broadcaster has a broadcast \val{} bounded by $b+n$. %
Since every such gadget has exactly one broadcaster by the previous arguments, 
it follows that the broadcast \val{} of any gadget is bounded by $b+n$. %
Furthermore, for all $j\in \bra{k}$, $j\neq i$, every vertex $E'_{ij}\setminus \{c_{ij}\}$ is at distance at most $2b+2n+1$ 
from $v$, so they cannot have a broadcast \val{} superior to $b+n-1$, which is also true for $c_{ij}$ as it is at distance $b+n+1$ from $f_{x_1}$. 
Once again, since every such gadget has exactly one broadcaster by the previous arguments, 
it follows that $\brp(E'_{ij}) \leq b+n-1$. 
Thus, the overall broadcast \val{} is bounded by (we let $J = \bra{k} \setminus \{i\}$): %
    \begin{align*}
        \brp(V_2) &= \brp(F) + \brp(v) + \sum\limits_{j \in J} \brp(I'_j) + 
        \sum\limits_{j,l \in J} \brp(E'_{jl}) + \sum\limits_{j \in J} \brp(E'_{ij}) \\
        & \leq a+1 + (b+n+1)+(k-1+ \frac{(k-1)(k-2)}{2})(b+n) + (k-1)(b+n-1) \\
        & < a+1 + (k+ \frac{k(k-1)}{2}) \cdot (b+n) \\
        & = M,
    \end{align*}
which is impossible. 

\medskip

For similar reasons, it is impossible for an edge gadget $E'_{ij}$ to have broadcast \val{} $b+n+1$: each independent set and edge gadget has thus a broadcast \val{} at most $b+n$. %
Therefore, equality must hold: each independent set and edge gadget has broadcast \val{} exactly $b+n$. 
\end{claimproof}

Note that for all $i\in \bra{k}$, $d(l_i,f_{x_1})=d(r_i,f_{x_1})= b+n+1$ and therefore $\brp(l_i),\brp(r_i)<b+n$, hence \cref{cl:bound} implies that they are not broadcasters in $I'_i$.  
A similar argument holds for all $i,j\in \bra{k}$ $i<j$ regarding $c_{ij}$.

For all $i,j\in\bra{k}$, $i<j$, let $v_i^{m_i}$ and $v_j^{m_j}$ be the unique broadcasters of $I'_i$ and $I'_j$, respectively. 
Moreover, let $v_{e_{ij}}$ be the only broadcaster of $E'_{ij}$, with broadcast \val{} $b+n$. 
The following holds:
\[d(v_i^{m_i},v_{e_{ij}})>\brp(v_i^{m_i})+\brp(v_{e_{ij}})=2b+2n\]

and symmetrically, $d(v_j^{m_j},v_{e_{ij}})>2b+2n$. 
But, by construction of $G_2$, this only holds if $e_{ij}=(v_i^{m_i},v_j^{m_j})$: 
indeed, if we consider $v_i^{m'}$ with $m' < m_i$ then: %
\begin{align*}
    d(v_i^{m'}, v_{e_{ij}}) & = d(v_i^{m'}, l_i) + d(l_i, v_{e_{ij}}) \\ 
    & = m' + 1 + 2b + 2n - m_i \\
    & \leq 2b + 2n. 
\end{align*}
A similar argument holds for $m' > m_i$ by going through $r_i$ instead. 
It follows that 
for all $i,j\in\bra{k}$ , $i<j$, there is an edge between $v_i^{m_i}$ and $v_j^{m_j}$ in $G_1$: $\{v_1^{m_1},\dots,v_k^{m_k}\}$ is a multicolored clique of $G_1$. 

\medskip

This concludes the proof of \cref{le:bptomclique}. 
\end{proofE}

\ThmWBPVCHard*

\begin{proofE}
\Cref{le:bptomclique} proves that $G_1$ is a yes instance of $k$-multicolored clique if and only if $(G_2,M)$ is a yes instance of \BP{}. 
Finally, note that the set 
\[
 X = \{s\} \cup \{l_i,r_i \mid i \in \bra{k}\} \cup \{c_{ij} \mid i,j\in \bra{k}, i<j \} 
 \]
is a vertex cover of $G_2$ of size $1 + 2k + \tfrac{k(k-1)}{2}$, yielding the claimed result.
\end{proofE}

Somewhat surprisingly, the edge-subdivision technique that works for \BI{} does not seem to extend to the unweighted version of \BP{}, as avoiding broadcasting on the newly created paths appears to be more difficult in this setting.
We leave this as an open question.

\section{Approximation algorithm parameterized by treewidth}
\label{sec:approx}

We complement our W[1]-hardness result for \BI{} (\cref{thm:BIFVSPWhard}) by providing a parameterized approximation algorithm 
within a ratio of $\frac{1}{2} - \epsilon$ for any $\epsilon > 0$. 
We will achieve this through a structural result that bounds 
the size on any independent $p$-broadcast with respect to that of a broadcast independent. 
For $p=1$ our result generalizes a result by \authorcite{DBLP:journals/dm/BessyR19} who showed 
that the size of a maximum independent broadcast is at most $4$ times the size of a maximum independent set 
(which corresponds to a maximum independent $1$-broadcast by definition). 

\apxresult{theorem}
{thm:bistruct}
{
    Let $G=(V,E)$ be a connected graph. If $G$ admits an independent broadcast $\bri$, then it admits an independent $p$-broadcast $\bri'$, where $\bri'(V)\geq \frac{p}{2p+2}\bri(V)$.
}
{approx}

\begin{proofE}
    Let $G=(V,E)$ be a connected graph and $\bri$ an independent broadcast of $G$. Intuitively, we construct an independent $p$-broadcast $\bri'$ as follows: we will keep all of the broadcasters of $\bri$ of \val{} at most $p$, truncate to value $p$ the broadcasters of \val{} between $p+1$ and $2p+2$, and replace each broadcaster of larger \val{} by several broadcasters placed on a shortest path starting at the original broadcaster (see \cref{fig:approx_indep}).%
    More formally, we construct $\bri'$ in the following way. For all $v\in V$ such that $\bri(v)>0$:
    \begin{itemize}
        \item if $\bri(v)\leq 2p+2$, then set $\bri'(v)=\min(\bri(v),p)$;
        \item if $\bri(v)> 2p+2$, then, as $\bri(v)\leq \ecc(v)$, there exists a vertex $u$ at distance $\beta(v) := \lfloor \dfrac{(\bri(v)+1)}{2} \rfloor$ from $v$. 
        Let $u_0^v,u_1^v,\dots,u_{\beta(v)}^v$ be a shortest path from $v$ to $u$ (with $v=u_0^v$ and $u=u_{\beta(v)}^v$). 
        Let $a, b \in \mathbb{N}$ such that %
        $\beta(v) = a(p+1)+b$, with $0 \leq b \leq p$. 
        Set $\bri'(u_{k(p+1)}^v):=p$ for all $k\in \{0,\dots, a-1\}$. %
        Then, if $b\geq 1$, set $\bri'(u_{a(p+1)}^v):=b$; 
        \item for every vertex $w\in V$ where $\bri'$ is still undefined, set $\bri'(w)=0$. 
    \end{itemize}
    
    By construction, $0 \le \bri'(w) \le p$ for all $w \in V$, hence $\bri'$ is a $p$-broadcast. 
    
    \begin{figure}[ht!]
\centering
\begin{tikzpicture}[
  v/.style  = {circle, draw=black, fill=white, inner sep=0pt, minimum size=6pt},
  bv/.style = {circle, draw=black, fill=black, inner sep=0pt, minimum size=8pt},
  >=Stealth,
  font=\small
]

\def\step{0.65}
\def\xcenter{4.88}

\foreach \i/\st in {0/bv,1/v,2/v,3/v,4/v,5/v,6/v,7/v,8/v,9/v,10/v,11/v,12/v,13/v,14/v,15/v}{
  \node[\st] (T\i) at (\i*\step, 0) {};
}
\foreach \i in {0,...,14}{
  \pgfmathtruncatemacro{\j}{\i+1}
  \draw (T\i) -- (T\j);
}
\node[below=3pt] at (T0) {$v$};
\node[below=3pt] at (T15) {$u^v_{15}$};

\draw[very thick, blue!75!black, ->, bend left=14]
  (T0) to node[above=5pt, midway, blue!75!black] {} (T15);

\draw[->, thick, gray!55] (\xcenter, -0.55) -- (\xcenter, -1.45);

\foreach \i/\st in {0/bv,1/v,2/v,3/bv,4/v,5/v,6/bv,7/v,8/v,9/v,10/v,11/v,12/v,13/v,14/v,15/v}{
  \node[\st] (B\i) at (\i*\step, -2.45) {};
}
\foreach \i in {0,...,14}{
  \pgfmathtruncatemacro{\j}{\i+1}
  \draw (B\i) -- (B\j);
}

\draw[fill=red!20, draw=red!70!black, line width=0.8pt, opacity=0.42]
  (B0) ellipse [x radius={1.45}, y radius=0.65];
\draw[fill=blue!20, draw=blue!70!black, line width=0.8pt, opacity=0.42]
  (B3) ellipse [x radius={1.45}, y radius=0.65];
\draw[fill=violet!20, draw=violet!60!black, line width=0.8pt, opacity=0.46]
  (B6) ellipse [x radius={1.45}, y radius=0.65];

\node[below=6pt] at (B0)  {$v$};
\node[below=6pt] at (B3)  {$u^v_{3}$};
\node[below=6pt] at (B6)  {$u^v_{6}$};
\node[below=6pt] at (B8)  {$u^v_{8}$};
\node[below=6pt] at (B15) {$u^v_{15}$};

\node[red!70!black, font=\footnotesize] at ($(B0)+(0,0.88)$) {$p$};
\node[blue!70!black, font=\footnotesize] at ($(B3)+(0,0.88)$) {$p$};
\node[violet!70!black, font=\footnotesize] at ($(B6)+(0,0.88)$) {$b$};

\end{tikzpicture}

\caption{\label{fig:approx_indep}
Example of a broadcast $\bri(v)=15$ transformed into an independent $p$-broadcast $\bri'$, with $p=2$. 
As $\beta(v) = \lfloor\frac{\bri(v)+1}{2}\rfloor = 8 = 2\cdot 3 + 2$, we have $a=b=2$, and $\bri'$ is of value $2\cdot a + b = 6 \geq \frac{2}{6}\bri(v)=5$.
}

\end{figure}
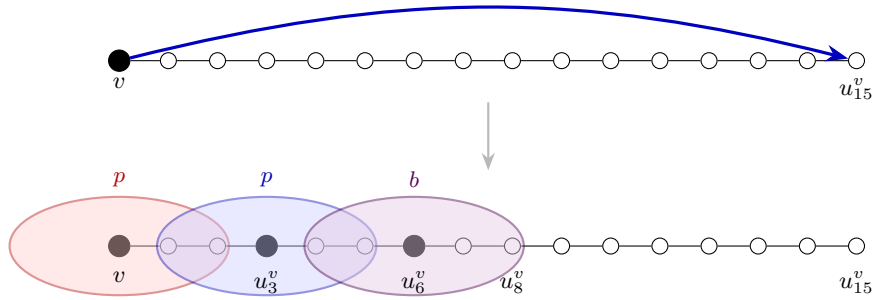
    
    \medskip
    
    \begin{claim}
    \label{cl:subballapx}
        The broadcast $\bri'$ is independent. 
    \end{claim}
    
    \begin{claimproof}
    For every original broadcaster $v$ such that $\bri(v)> 2p+2$, note that by construction, for every $u_k^v$ s.t. $\bri'(u_k^v)>0$, we have: %
    \begin{equation}
        \label{eq:subball2apx}
        d(v,u^v_k) + \bri'(u^v_k) \le \beta(v). %
    \end{equation}

    Let $u_1,u_2\in V$ be such that $u_1\neq u_2$ and $\bri'(u_1)>0$ and $\bri'(u_2)>0$. Let us show that $\bri'(u_1)<d(u_1,u_2)$ and $\bri'(u_2)<d(u_1,u_2)$:
    \begin{itemize}
        \item if $\bri(u_1)>0$ and $\bri(u_2)>0$, then $\bri'(u_1) \leqslant \bri(u_1)<d(u_1,u_2)$ by construction, and because $\bri$ is independent. The same holds for $\bri'(u_2)$ symmetrically;
        
        \item if $\bri(u_1)>0$ and $\bri(u_2)=0$, then there exist $v\in V$ and $k\in \bra{\beta(v)-1}$ such that $\bri(v)> 2p+2$ and $u_2=u_k^v$. 
        
        If $v \neq u_1$ then, $d(v,u_1)>\bri(v)\geq \beta(v) + p+1$ because $\bri$ is independent and $\bri(v)> 2p+2$, $d(v,u_2)+\bri'(u_2) \leq \beta(v)$ by \cref{eq:subball2apx}, and $\bri'(u_1)\leq p$ by construction. 
        
        We obtain that $\bri'(u_1)\leq p < d(v,u_1)-\beta(v) \leq d(v,u_1)-d(v,u_2)-\bri'(u_2) \leq d(v,u_1)-d(v,u_2)$. By triangular inequality, we also have $d(v, u1) \leq d(v, u_2) + d(u_2, u_1)$. Then, $\bri'(u_1) < d(u_1,u_2)$. %
        Similarly, in the other direction, we get $\bri'(u_2)\leq \beta(v) - d(v,u_2)< d(v,u_1)-d(v,u_2)\leq d(u_1,u_2)$. 
        
        If $v = u_1$, $d(u_1,u_2) > \bri'(u_1) + \bri'(u_2)$ by construction. %
        The same result holds symmetrically if $\bri(u_1)=0$ and $\bri(u_2)>0$;
        
        \item if $\bri(u_1)=0$ and $\bri(u_2)=0$, then there exist $v_1,v_2\in V$ and $k_1\in \{1,\beta(v_1)-1\}$, $k_2\in \{1,\beta(v_2)-1\}$ such that $\bri(v_1)> 2p+2$, $u_1=u_{k_1}^{v_1}$, $\bri(v_2)> 2p+2$, and $u_2=u_{k_2}^{v_2}$.
        
        If $v_1\neq v_2$, $\bri(v_1)<d(v_1,v_2)$ and $\bri(v_2)<d(v_1,v_2)$ because $\bri$ is independent; furthermore, $d(v_1,u_1)+\bri'(u_1)\leq \beta(v_1)$ and $d(v_2,u_2)+\bri'(u_2)\leq \beta(v_2)$ by \cref{eq:subball2apx}.
        
        We obtain that $\bri'(u_1)\leq \beta(v_1) - d(v_1,u_1)$. Furthermore, $\bri(v_1)<d(v_1,v_2)$, so $\beta(v_1) \leq \dfrac{d(v_1,v_2)}{2}$ and $\bri'(u_1)\leq \dfrac{d(v_1,v_2)}{2}-d(u_1,v_1)$. Symmetrically, $\bri'(u_2)\leq \dfrac{d(v_1,v_2)}{2}-d(u_2,v_2)$. As $\bri'(u_2)>0$, we have $\dfrac{d(v_1,v_2)}{2}-d(u_2,v_2)>0$. 
        
        Therefore, $\bri'(u_1)< \dfrac{d(v_1,v_2)}{2}-d(u_1,v_1)+\dfrac{d(v_1,v_2)}{2}-d(u_2,v_2)$. %
        Using two triangular inequalities, $d(v_1, v_2) \leq d(v_1, u_1)+d(u_1, v_2)$ and $d(u_1, v_2) \leq d(u_1, u_2)+d(u_2, v_2)$, we can conclude that $\bri'(u_1)< d(v_1, v_2)-d(u_1, v_1)-d(u_2, v_2) \leq d(u_1,u_2)$. The same holds symmetrically for $\bri'(u_2)$.
        
        If $v_1=v_2$, we have $d(u_1,u_2)\geq p+1$ and $\bri'(u_1),\bri'(u_2)\leq p$ by construction, which is the promised result.
    \end{itemize}
    \smallskip
    This concludes the proof of \cref{cl:subballapx}. 
    \end{claimproof}
    
    By \cref{cl:subballapx}, $\bri'$ is indeed an independent $p$-broadcast. 
    It remains to bound its \val. For each original broadcaster $v$, we define $\omega(v)$ as such:
    \[
    \omega(v) = 
    \begin{cases}
        \min(\bri(v), p) & \mid \bri(v) \leq 2p+2 \\
        \sum\limits_{i \in [\bri(v)]} \bri'(u^v_i) & \mid \bri(v) > 2p+2.
    \end{cases}
    \]
    If $\bri(v)\leq 2p+2$, then $\omega(v)=\min(\bri(v),p)\geq \frac{p}{2p+2}\bri(v)$. 
    If $\bri(v)> 2p+2$, there exist $a, b \in \mathbb{N}$ such that $\beta(v) = a(p+1)+b$, with $0 \leq b \leq p$. %
    Then, by construction, $\omega(v)=a\cdot p + b$, and $\omega(v)\geq\frac{p}{p+1}\beta(v)\geq\frac{p}{2p+2}\bri(v)$. 
    We then sum over the original broadcasters and obtain $\bri'(V)\geq\frac{p}{2p+2}\bri(V)$, as claimed.
\end{proofE}

\ThmBIapprox*

\begin{proofE}
    Let $G=(V,E)$ be a graph, $\epsilon>0$, and $\bri_1$ an optimal independent broadcast. Let $p:= \lfloor \frac{1}{2\epsilon}\rfloor$. From \cref{thm:WBPItw}, we can compute an optimal independent $p$-broadcast of $G$ in time $(p+1)^{O(\tw)}n^{O(1)}$, i.e. in time $(\frac{1}{\epsilon})^{O(\tw)}n^{O(1)}$. Let $\bri_2$ be such a $p$-broadcast. Furthermore, from \cref{thm:bistruct}, there is an independent $p$-broadcast of value at least $\frac{p}{2p+2}\bri_1(V)$. As $\bri_2$ is an optimal independent $p$-broadcast, 
    $\bri_2$ satisfies this inequality and has value at least $\frac{p}{2p+2}\bri_1(V)$, i.e. at least $(\frac{1}{2}-\epsilon) \bri_1(V)$.
\end{proofE}

We finally present a similar structural result for \BP. Note that an approximation algorithm can also be deduced from it. 
However, we were not able to prove that the problem is W[1]-hard parameterized by treewidth. 

\apxresult{theorem}
{thm:bpstruct}
{
    Let $G=(V,E)$ be a connected graph. If $G$ admits a broadcast packing $\brp$, then it admits a $p$-broadcast packing $\brp'$, where $\brp'(V)\geq \frac{p}{2p+1}\brp(V)$.
}
{approx}

\begin{proofE}
    Let $G=(V,E)$ be a simple graph, and $\brp$ a broadcast packing of $G$. Intuitively, we construct a $p$-broadcast packing $\brp'$ as follows: we will keep all the broadcasters of $\brp$ of \val{} at most $p$, and replace each broadcaster of \val{} larger than $p$ by several broadcasters on a shortest path having the original broadcaster as an endpoint (see \cref{fig:approx_packing}). %
    
    We construct $\brp'$ in the following way.
    For all $v\in V$ such that $\brp(v)>0$:
    \begin{itemize}
        \item if $\brp(v)\leq p$, then set $\brp'(v)= \brp(v)$;%
        \item if $\brp(v)> p$, then, as $\brp(v)\leq \ecc(v)$, there exists a vertex $u$ at distance $\brp(v)$ from $v$. 
        Let $u_0^v, u_1^v,\dots, u_{\brp(v)}^v$ be a shortest path from $v$ to such a $u$ (with $v = u_0^v$ and $u=u_{\brp(v)}^v$).
        Since $\brp$ is a packing, all vertices $u_1^v,\dots, u_{\brp(v)}^v$ are non-broadcasters in $\brp$.
        
        Let $a, b \in \mathbb{N}$ such that $\brp(v)-p = a(2p+1)+b$, with $a \geqslant 0$ %
        and $0 \le b \le 2p$.
        For every $k\in \braz{a}$, set $\brp'\left(u^v_{k(2p+1)}\right)=p$.
        If $b \ge 3$, set $\brp'\left(u_{p+a(2p+1)+\lceil b/2\rceil}^v\right)= \lceil b/2\rceil -1$; 
        \item for every vertex $w\in V$ where $\brp'$ is still undefined, set $\brp'(w):=0$. %
    \end{itemize}
    
    By construction, $0 \le \brp'(w) \le p$ for all $w \in V$, hence $\brp'$ is a $p$-broadcast.
    
    \begin{figure}[ht!]
\centering
\begin{tikzpicture}[
  v/.style  = {circle, draw=black, fill=white, inner sep=0pt, minimum size=6pt},
  bv/.style = {circle, draw=black, fill=black, inner sep=0pt, minimum size=8pt},
  >=Stealth,
  font=\small
]

\def\step{0.65}
\def\xcenter{4.88}

\foreach \i/\st in {0/bv,1/v,2/v,3/v,4/v,5/v,6/v,7/v,8/v,9/v,10/v,11/v,12/v,13/v,14/v,15/v}{
  \node[\st] (T\i) at (\i*\step, 0) {};
}
\foreach \i in {0,...,14}{
  \pgfmathtruncatemacro{\j}{\i+1}
  \draw (T\i) -- (T\j);
}
\node[below=3pt] at (T0) {$v$};
\node[below=3pt] at (T15) {$u^v_{15}$};

\draw[very thick, blue!75!black, ->, bend left=14]
  (T0) to node[above=5pt, midway, blue!75!black] {} (T15);

\draw[->, thick, gray!55] (\xcenter, -0.55) -- (\xcenter, -1.45);

\foreach \i/\st in {0/bv,1/v,2/v,3/v,4/v,5/bv,6/v,7/v,8/v,9/v,10/bv,11/v,12/v,13/v,14/bv,15/v}{
  \node[\st] (B\i) at (\i*\step, -2.45) {};
}
\foreach \i in {0,...,14}{
  \pgfmathtruncatemacro{\j}{\i+1}
  \draw (B\i) -- (B\j);
}

\draw[fill=red!20, draw=red!70!black, line width=0.8pt, opacity=0.42]
  (B0) ellipse [x radius={1.45}, y radius=0.65];
\draw[fill=blue!20, draw=blue!70!black, line width=0.8pt, opacity=0.42]
  (B5) ellipse [x radius={1.45}, y radius=0.65];
\draw[fill=blue!20, draw=blue!70!black, line width=0.8pt, opacity=0.42]
  (B10) ellipse [x radius={1.45}, y radius=0.65];
\draw[fill=violet!20, draw=violet!60!black, line width=0.8pt, opacity=0.46]
  (B14) ellipse [x radius={0.80}, y radius=0.30];

\node[below=6pt] at (B0)  {$v$};
\node[below=6pt] at (B5)  {$u^v_5$};
\node[below=6pt] at (B10) {$u^v_{10}$};
\node[below=6pt] at (B14) {$u^v_{14}$};
\node[below=6pt] at (B15) {$u^v_{15}$};

\node[red!70!black, font=\footnotesize] at ($(B0)+(0,0.88)$) {$p$};
\node[blue!70!black, font=\footnotesize] at ($(B5)+(0,0.88)$) {$p$};
\node[blue!70!black, font=\footnotesize] at ($(B10)+(0,0.88)$) {$p$};
\node[violet!60!black, font=\footnotesize] at ($(B14)+(0,0.88)$) {$\lceil b/2\rceil - 1$};%

\draw[decorate, decoration={brace, amplitude=5pt, mirror}, red!70!black, line width=1pt]
  ($(B0.south west)+(0,-0.62)$) -- ($(B2.south east)+(-0.06,-0.62)$)
  node[midway, below=7pt, font=\footnotesize, red!70!black] {$p$};

\draw[decorate, decoration={brace, amplitude=5pt, mirror}, blue!70!black, line width=1pt]
  ($(B2.south west)+(0.06,-0.62)$) -- ($(B12.south east)+(-0.06,-0.62)$)
  node[midway, below=7pt, font=\footnotesize, blue!70!black] {$a(2p+1)$};

\draw[decorate, decoration={brace, amplitude=5pt, mirror}, violet!60!black, line width=1pt]
  ($(B12.south west)+(0.06,-0.62)$) -- ($(B15.south east)+(0,-0.62)$)
  node[midway, below=7pt, font=\footnotesize, violet!60!black] {$b$};

\end{tikzpicture}

\caption{\label{fig:approx_packing}
Example of a broadcast $\brp(v)=15$ transformed into a $p$-broadcast packing $\brp'$, with $p=2$. 
As $\brp(v)-p = 2\cdot 5 + 3$, we have $a = 2$ and $b=3$, and
$\brp'$ is of value $2 + 2\cdot a + \lceil b/2\rceil - 1 = 7 \geq \frac{2}{5}\brp(v)=6$.
}

\end{figure}
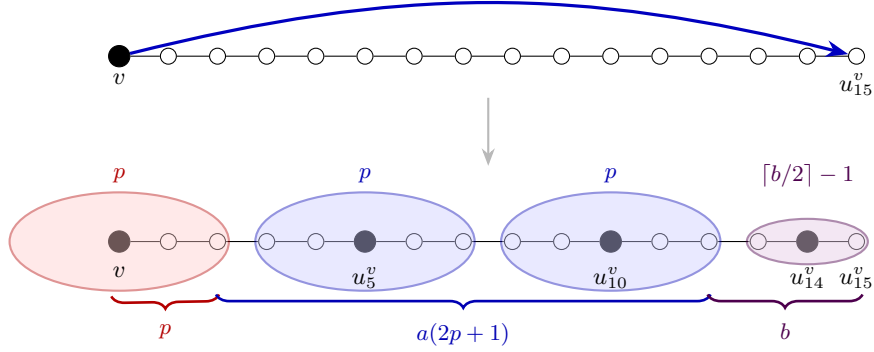

    \begin{claim}
    \label{cl:subball}
        The broadcast $\brp'$ is a packing.
    \end{claim}
    
    \begin{claimproof}
    First, note that by construction, for every original broadcaster $v$ with $\brp(v)>p$ and every $u^v_k$ s.t. $\brp'(u^v_k)>0$, we have
    \begin{equation}
        \label{eq:subball}
        d(v,u^v_k) + \brp'(u^v_k) \le \brp(v).
    \end{equation}
    
    Let $u_1,u_2$ be such that $u_1\neq u_2$ and $\brp'(u_1)>0$, $\brp'(u_2)>0$. 
    Then: 
    \begin{itemize}
        \item if $\brp(u_1)> 0$ and $\brp(u_2)>0$, then $\brp'(u_1)+\brp'(u_2)\leq \brp(u_1)+\brp(u_2)<d(u_1,u_2)$, by construction of $\brp'$ and because $\brp$ is a packing;
        
        \item if $\brp(u_1)>0$ and $\brp(u_2)=0$, then there exist $v\in V$ and $k\in  \bra{\brp(v)}$ such that $u_2=u_k^v$. %
        If $v \neq u_1$, then both $u_1$ and $v$ are distinct original broadcasters in $\brp$, 
        so the packing condition gives $d(u_1,v) > \brp(u_1) + \brp(v)$.
        Moreover, by \cref{eq:subball}, $d(v,u_2)+\brp'(u_2) \leq \brp(v)$ and by construction, $\brp(u_1) \geq \brp'(u_1)$. Therefore, by the triangle inequality, we have $d(u_1,v) \leqslant d(u_1, u_2) + d(u_2, v)$ and hence: %
        \begin{align*}
                d(u_1,u_2) & \geq d(u_1,v) - d(v,u_2) \\ 
                & > \brp(u_1) + \brp(v) - (\brp(v) - \brp'(u_2)) \\
                & > \brp(u_1) + \brp'(u_2) \\
                & > \brp'(u_1) + \brp'(u_2) 
        \end{align*}
        
        as desired. If $v = u_1$, $d(u_1,u_2) > \brp'(u_1) + \brp'(u_2)$ by construction. 
        The proof is symmetric if $\brp(u_1)=0$ and $\brp(u_2)>0$;
        
        \item if $\brp(u_1)=\brp(u_2)=0$, then there are $v_1,v_2\in V$, $k_1\in \bra{\brp(v_1)}$, $k_2\in \bra{\brp(v_2)}$, such that $u_1=u_{k_1}^{v_1}$ and $u_2=u_{k_2}^{v_2}$. 
        If $v_1\neq v_2$, as $\brp$ is a packing, $\brp(v_1) + \brp(v_2) < d(v_1,v_2)$. 
        Moreover, by \cref{eq:subball}, $d(u_1,v_1)+\brp'(u_1)\leq \brp(v_1)$, and $d(u_2,v_2)+\brp'(u_2)\leq \brp(v_2)$.
        Therefore, $\brp'(u_1)+\brp'(u_2)<d(v_1,v_2)-d(u_1,v_1)-d(u_2,v_2)\leq d(u_1,u_2)$, by the triangle inequality.
        Otherwise, if $v_1=v_2$, then $\brp'(u_1)+\brp'(u_2)< d(u_1,u_2)$ by construction.
        
    \end{itemize}
    This concludes the proof of \cref{cl:subball}. 
    \end{claimproof}
    
    By \cref{cl:subball}, $\brp'$ is indeed a $p$-broadcast packing. 
    It remains to bound its \val.
    For each original broadcaster $v$, we define $\omega(v)$ as such:
    \[
    \omega(v) = 
    \begin{cases}
        \brp(v) & \mid \brp(v) \leq p \\
        \sum\limits_{i \in [\brp(v)]} \brp'(u^v_i) & \mid \brp(v) > p.
    \end{cases}
    \]
    If $\brp(v)\leq p$, then $\omega(v) = \brp(v) \geq \frac{p}{2p+1}\brp(v)$. 
    If $\brp(v) > p$, there exist $a, b \in \mathbb{N}$ such that $\brp(v) - p = a(2p+1) + b$, with $0 \leq b \leq 2p$.
    Then, by construction, $\omega(v) = p + a\cdot p + \max(0,\lceil b/2 \rceil -1)$. 
    We want to show that $(2p+1) \cdot w(v)\ge p\cdot \brp(v)$, that is:
    \begin{align*}
         (2p+1) \cdot \bigl(p+ap+\max\{0,\lceil b/2\rceil-1\}\bigr)&\ge p \cdot \bigl(p+a(2p+1)+b\bigr) \\
        \Leftrightarrow 
         (2p+1) \cdot \max\{0,\lceil b/2\rceil-1\}&\ge p \cdot (b-p-1).
    \end{align*}
    
    Which is true for all $b \in [2p]$. %
    Summing over all original broadcasters, we obtain $\brp'(V) \geq \frac{p}{2p+1}\brp(V)$ as claimed.
\end{proofE}

\section{Conclusion}
In this work, we initiated the study of the parameterized complexity of \BI{} and \BP{} for standard structural parameters. Prior to our results it was only known that \BI{} is W[1]-hard parameterized by the natural parameter $k$ %
through an easy reduction from \textsc{Independent Set}. 
Using dynamic programming on nice tree decompositions, we obtained FPT algorithms parameterized by 
$\tw$ plus diameter. 
Our result generalizes the known polynomial-time algorithm for trees by \authorcite{DBLP:journals/dam/BessyR22}. 
We complemented this result by providing hardness for several parameterizations, namely vertex cover for weighted variants of both problems and feedback vertex set size plus pathwidth for \BI. 
Together, these results imply hardness parameterized by treewidth for \BI. We moreover provided constant factor approximation algorithms parameterized by treewidth, thus %
completing the picture for this problem-parameter combination. 

The first natural open problem that arises from our work is to determine the complexity of \BP{} parameterized by treewidth. Note that if it turns out to be W[1]-hard then \cref{thm:bpstruct} would provide a constant-factor parameterized approximation algorithm as well. 
Finally, we believe that our algorithms can be extended to graphs of bounded cliquewidth and generalized to other metric-distance problems. We leave this as an open problem. %

\bibliographystyle{plainurl}%
\bibliography{bib}

\end{document}